\documentclass[a4paper,fleqn]{cas-dc}

\usepackage[authoryear]{natbib}
\usepackage[normalem]{ulem}
\usepackage{amsmath,amssymb,amsfonts,amsthm}
\usepackage{algorithmic}
\usepackage{graphicx}
\usepackage{subfigure}
\usepackage{caption}
\usepackage{textcomp}

\newtheorem{theorem}{Theorem}
\newtheorem{lemma}{Lemma}
\newtheorem{proposition}{Proposition}
\newtheorem{assumption}{Assumption}
\newtheorem{remark}{Remark}
\newtheorem{example}{Example}
\newtheorem{corollary}{Corollary}
\newtheorem{definition}{Definition}
 \newdefinition{rmk}{Remark}

\def\tsc#1{\csdef{#1}{\textsc{\lowercase{#1}}\xspace}}
\tsc{WGM}
\tsc{QE}
\tsc{EP}
\tsc{PMS}
\tsc{BEC}
\tsc{DE}

\begin{document}

\let\WriteBookmarks\relax
\def\floatpagepagefraction{1}
\def\textpagefraction{.001}
\shorttitle{Resource Availability in the Social Cloud: An Economics Perspective}
\shortauthors{P.~C.~Mane ~et~al.}

\title [mode = title]{Resource Availability in the Social Cloud: An Economics Perspective}

\author[1]{Pramod C. Mane}[]
\ead{pmane.cs@nitrr.ac.in}
\address[1]{Department of Computer Science and Engineering, National Institute of Technology Raipur, Raipur, India.}

\author[2]{Nagarajan Krishnamurthy}[]
\address[2]{Operations Management and Quantitative Techniques, Indian Institute of Management Indore, Indore, India.}
\ead{nagarajan@iimidr.ac.in}

\author[3]{Kapil Ahuja}[]
\address[3]{Department of Computer Science and Engineering, Indian Institute of Technology Indore, Indore, India.}
\ead{kahuja@iiti.ac.in}

\sffamily

\begin{abstract}
This paper focuses on social cloud formation, where agents are involved in a closeness-based conditional resource sharing and build their resource sharing network themselves. The objectives of this paper are: (1) to investigate the impact of agents' decisions of link addition and deletion on their local and global resource availability, (2) to analyze spillover effects in terms of the impact of link addition between a pair of agents on others' utility, (3) to study the role of agents' closeness in determining what type of spillover effects these agents experience in the network, and (4) to model the choices of agents that suggest with whom they want to add links in the social cloud. The findings include the following. Firstly, agents' decision of link addition (deletion) increases (decreases) their local resource availability. However, these observations do not hold in the case of global resource availability. Secondly, in a connected network, agents experience either positive or negative spillover effect and there is no case with no spillover effects. Agents observe no spillover effects if and only if the network is disconnected and consists of more than two components (sub-networks). Furthermore, if there is no change in the closeness of an agent (not involved in link addition) due to a newly added link, then the agent experiences negative spillover effect. Although an increase in the closeness of agents is necessary in order to experience positive spillover effects, the condition is not sufficient. By focusing on parameters such as closeness and shortest distances, we provide conditions under which agents choose to add links so as to maximise their resource availability. 
\end{abstract}


\maketitle

\section{Introduction}

The idea of the Social Cloud \citep{transaction, conference, Chard-Clouds-Retrospective-2015} has received much attention in the last few years. These systems take advantage of social connections to offer a secure and reliable way of resource sharing between agents. Researchers believe that exploiting social connections (in the form of social networks) can aid dealing with various issues like resource sharing policies and mechanisms \citep{socialcloud-as-communitycloud, joballocationscoialcloud}, trust \citep{trustfoundations, distributed}, and incentivising resource sharing \citep{incentivesocial, cooperativesocialcloud}. In this context, social connections are either exogenous or endogenous. Exogenous social connections are those social connections which are extracted from an online social network (for example, Facebook) in the form of a social graph. Whereas endogenous social connections are those social connections which are constructed in the context of social cloud application (for example, BuddyBackup\footnote{http://www.buddybackup.com/(Visited on 01 Jan 2021).}), where agents are decision makers who build their resource sharing connections.

Recent research trends in social cloud have led to two directions. One trend focuses on the role of exogenous social connections in defining quality of services (for example, data availability, reliability), and trust in social cloud. For example, \cite{Zuo-2016-P2P} show that a small set of friends play a crucial role in deterring the quality of services and workload balance. Another research trend \citep{Pramod-ANOR, Pramod-Games, social-cloud-gamenets, Pramod-AEL, Stefan2011} focuses on endogenous social connections in terms of resource sharing network formation, the stability and efficiency analysis of these networks, and analysis of externalities in these resource sharing networks. For example, \cite{Pramod-ANOR} study social storage cloud formation in a strategic setting, where self-interested agents build a storage resource sharing network for maximizing their respective utilities. They show that for the given degree-based utility of agents in this social cloud setting, agents always form the $\eta$-regular network, where each agent has $\eta$ neighbors.

However, the present literature on social cloud has left behind several aspects, for example, 1) the impact of link addition and deletion on resource availability of those who are involved in the link addition and deletion, 2) the impact of link addition between a pair of agents on other agents' utilities, 3) choice modelling that captures with whom self-interested agents add new social connections and delete existing social connections. This paper aims to fill these gaps. 

Following are the objectives of this study. The first objective is to study the impact of agents' decision of link addition and deletion on their local as well as global resource availability. We show that for the utility model proposed in \citep{social-cloud-gamenets}, the local resource availability of a pair of agents increases by adding a link and decreases by deleting the link between them. However, in the case of global resource availability, the same is not true. The second objective is to analyze agents' local resource availability in the context of their local connections. We find that, in the case of link addition, agents never observe an increase in their local resource availability from their neighbors. The opposite is true in the case of link deletion. The third objective is to analyze the impact of link addition between a pair of agents on the utility of others. This aspect can be outlined in terms of the spillover effect. \cite{Pramod-Games} provide a necessary and sufficient condition under which an agent experiences positive or negative spillover effect. In \cite{Pramod-AEL}, the authors follow an empirical approach to study the role of network structure and size in determining spillover. In this paper, we throw attention on the role of agents' closeness in determining the kind of spillover effects they experience due to a newly added link in the network. We show that an increase in closeness is necessary, but not sufficient, for an agent to experience positive spillover effect. We show that in the two diameter network, agents always experience negative spillover effect. The fourth objective is to understand the preferences of agents in link addition.  That is, with whom agents prefer to add links in a network. For this, we provide a set of conditions by taking the distances and closeness of agents into consideration.

\section{The Social Cloud Model}\label{sec:model}
For the sake of completeness, we first describe the social cloud model presented in \cite{social-cloud-gamenets}. A social cloud can be seen as a socially-aware resource sharing network $\mathfrak{g}=\{\mathbf{A}, \mathbf{L}\}$ that consists of a non-empty set $\mathbf{A}$ of $n$ agents and a set $\mathbf{L}$ of $\ell$ undirected links connecting these agents. One can view the set $\mathbf{L}$ as a platform that facilitates agents to share their computing resources, such as disk space and computing power, with others, and search for resources shared by others. An undirected link $\langle ij\rangle \in \mathbf{L}$ represents a direct communication channel between agents $i$ and $j$ in $\mathfrak{g}$. In other words, agents $i$ and $j$ are neighbours in $\mathfrak{g}$. The set $\eta_{i}(\mathfrak{g})$ represents the number of neighbors of agent $i$ in $\mathfrak{g}$.

A path in $\mathfrak{g}$ connecting agents $i_1$ and $i_n$ is a sequence of distinct agents $(i_1, i_2, . . ., i_n)$ such that $\langle i_1 i_2\rangle, \langle i_2 i_3\rangle, \ldots, \langle i_{n-1} i_n\rangle$ $\in \mathbf{L}$. The length of a path is the number of links that the path contains. A shortest path between agent $i$ and agent $j$ is a path between $i$ and $j$ that has the least length, among all paths between $i$ and $j$. The \textit{distance} between $i$ and $j$, $d_{ij}(\mathfrak{g})$, is the length of the shortest path between them. We say, $i$ and $j$ are $d_{ij}(\mathfrak{g})$ hops away from each other. The diameter, $\mathcal{D}_{\mathfrak{g}}$, of $\mathfrak{g}$ is the maximum distance between any pair of agents. The radius of $\mathfrak{g}$ is the minimum distance between any pair of agents. A network $\mathfrak{g}$ is connected if there exists at least one path between any pair of agents; otherwise it is disconnected. A disconnected network $\mathfrak{g}$ is a collection of two or moredisjoint components (sub-networks) $\mathfrak{g}(\mathfrak{c}_1), \mathfrak{g}(\mathfrak{c}_2) \cdots \mathfrak{g}(\mathfrak{c}_z)$ such that $\mathfrak{c}_1 \cup \mathfrak{c}_2 \cup \cdots \cup \mathfrak{c}_z=\mathbf{A}$, and $\mathfrak{c}_x \cap \mathfrak{c}_y= \emptyset$ for all $x, y \in \{1, 2, \cdots, z\}, x\not=y$, such that any pair of agents $i$ and $j$ are connected if and only if they are elements of the same set $\mathfrak{c}_x$.  

A network $\mathfrak{g}$ evolves when agents perform two operations, namely, (1) link addition, where agents $i, j$ in $\mathfrak{g}$, $\langle ij\rangle \not \in \mathfrak{g}$, mutually add the link $\langle ij\rangle$, resulting in the network $\mathfrak{g}+\langle ij\rangle$, and 
(2) link deletion, where agents $k, l$ in $\mathfrak{g}$, $\langle kl\rangle \in \mathfrak{g}$, unilaterally or mutually 
delete $\langle kl\rangle$ to give $\mathfrak{g}-\langle kl\rangle$. 

Henceforth, we refer to a socially-aware resource sharing network as a resource sharing network, and whenever we refer 
to $\mathfrak{g}$, we mean a resource sharing network.

\subsection{Assumptions}
In $\mathfrak{g}$, agents who have underutilised resources share the same with other agents who want to use the resources to accomplish their computational tasks. For example, an agent may need to backup its data and may use the storage space shared by another agent in $\mathfrak{g}$. Now, we state the basic assumptions on which the social cloud model stands.
\begin{assumption} 
Agents in $\mathfrak{g}$ share the same kind of resource. We denote the resource by $\mathfrak{r}$.
\end{assumption}

\begin{assumption} 
In the prevailing $\mathfrak{g}$, an agent has underutilised resource $\mathfrak{r}$ with probability $p$ and needs to perform a computational task using the resource with probability $q$.
\end{assumption}

\begin{assumption} 
An agent plays either the role of a resource provider or that of a resource consumer, with probabilities $p(1-q)$ and $q(1-p)$, respectively.
\end{assumption}
\begin{assumption}
Each agent has global information, that is, each agent is aware of the network structure $\mathfrak{g}$ and the prevailing resource sharing situation in $\mathfrak{g}$.
\end{assumption}

\subsection{Closeness-Based Resource Sharing}
In $\mathfrak{g}$, agents perform closeness based resource sharing \citep{transaction}. For example, agents could limit resource sharing with those agents who are close to them. The notion of how close an agent is to all other agents can be captured by the harmonic centrality measure, discussed in \citep{Boldi-axioms-centrality, opsahl, marchiori2000harmony}, defined as follows:
\begin{equation}\label{eq:harmonic-closeness}
\Phi_{i}(\mathfrak{g})=\sum\limits_{j \in \mathfrak{g} \setminus \{i\}} \frac{1}{d_{ij}(\mathfrak{g})}
\end{equation}
$\Phi_{i}(\mathfrak{g})$ is called the closeness of $i$ in $\mathfrak{g}$. Harmonic centrality handles $\infty$ smoothly, and hence, this centrality measure deals with disconnected networks too. 

In $\mathfrak{g}$, an agent $j \in \mathfrak{g}$ (who acts as a resource provider) computes a probability distribution on all agents for the purpose of allocating the resource to agent $i\in \mathfrak{g}$ (who acts as a resource consumer), as given below:
\begin{equation}\label{eq:agent-probability} 
\alpha_{ij}(\mathfrak{g})=p(1-q) \frac{\frac{1}{d_{ij}(\mathfrak{g})}}{\sum\limits_{j \in \mathfrak{g}\setminus \{i\}} \frac{1}{d_{ij}(\mathfrak{g})}}=\frac{p(1-q)}{{d_{ij}(\mathfrak{g})\Phi_{i}(\mathfrak{g})}}.
\end{equation}
In other words, $\alpha_{ij}(\mathfrak{g})$ is the probability that agent $i$ will obtain the resource from agent $j$ in $\mathfrak{g}$.
\begin{remark}
If $d_{ij}(\mathfrak{g})=\infty$ then $\alpha_{ij}(\mathfrak{g})=0$ ($= \alpha_{ji}(\mathfrak{g})$). As agents $i$ and $j$ are disconnected in $\mathfrak{g}$ their chances of obtaining the resource from each other is nil.
\end{remark}
The probability that agent $i$ obtains the resource from at least one agent in $\mathfrak{g}$ is as follows:
\begin{equation}\label{eq:gra} 
\gamma_{i}(\mathfrak{g})=1-\prod \limits_{j\in \mathfrak{g}\setminus\{i\}}(1-\alpha_{ij}(\mathfrak{g})).
\end{equation}

\begin{definition}
We call $\alpha_{ij}(\mathfrak{g})$ the {\it local resource availability} of $i$ from $j$ in $\mathfrak{g}$, and $\gamma_{i}(\mathfrak{g})$, the {\it global resource availability} of $i$ in $\mathfrak{g}$.
\end{definition}

An agent's chance of obtaining the resource from another agent is determined by, first, the distance between the agent (who wants the resource) and the other agent (who may provide the resource), and second, the other agent's closeness. Hence, it is important to get to know how a newly added link affects the distance between pairs of agents and their closeness.

\begin{remark}\label{rmk:conditions-remark}
Suppose $k, l$ are distinct agents in $\mathfrak{g}$ such that $\langle kl \rangle \notin \mathfrak{g}$. 
Then, 

either $d_{ij}(\mathfrak{g})=d_{ij}(\mathfrak{g}+\langle  kl\rangle)$, for all $i \in \mathfrak{g} \setminus \{k, l\}$, $j \in \mathfrak{g}$,  

or $d_{ij}(\mathfrak{g})=d_{ij}(\mathfrak{g}+\langle  kl\rangle)$ for some $i \in \mathfrak{g}\setminus\{k, l\}$, $j \in \mathfrak{g}$, 

\quad and $d_{ij}(\mathfrak{g})>d_{ij}(\mathfrak{g}+\langle  kl\rangle)$ for others. \\
Similarly,

either $\Phi_i(\mathfrak{g})=\Phi_i(\mathfrak{g}+\langle  kl\rangle)$, for all $i \in \mathfrak{g}\setminus\{k, l\}$, 

or 
$\Phi_i(\mathfrak{g})=\Phi_i(\mathfrak{g}+\langle  kl\rangle)$, for some $i \in \mathfrak{g}\setminus\{k, l\}$ 

\quad and $\Phi_i(\mathfrak{g})<\Phi_i(\mathfrak{g}+\langle  kl\rangle)$, for others. 
\end{remark}

\begin{lemma}\label{lemma:g=g+kl}
Suppose $i, j, k, l$ are distinct agents in $\mathfrak{g}$ such that $\langle ij \rangle \notin \mathfrak{g}$ and $\langle kl \rangle \notin \mathfrak{g}$. If $d_{ij}(\mathfrak{g})>d_{ij}(\mathfrak{g}+\langle  kl\rangle)$ then $\Phi_i(\mathfrak{g})<\Phi_i(\mathfrak{g}+\langle  kl\rangle)$ and $\Phi_j(\mathfrak{g})<\Phi_j(\mathfrak{g}+\langle  kl\rangle)$.
\end{lemma}
\begin{proof}

$d_{ij}(\mathfrak{g}+\langle  kl\rangle) \leq d_{ij}(\mathfrak{g})-1$. On adding $\langle  kl\rangle$, even if the distances of all other agents remain unchanged, the closeness of both $i$ and $j$ will increase by at least $\frac{1}{d_{ij}(\mathfrak{g})-1} - \frac{1}{d_{ij}(\mathfrak{g})}$, or $\frac{1}{d_{ij}(\mathfrak{g})(d_{ij}(\mathfrak{g})-1)}$.
\end{proof}

Due to Remark \ref{rmk:conditions-remark} and Lemma \ref{lemma:g=g+kl}, we study an agent's probability of obtaining the resource in $\mathfrak{g}$ by taking the following cases into the consideration. 
\begin{enumerate}
\item $d_{ij}(\mathfrak{g})=d_{ij}(\mathfrak{g}+\langle  kl\rangle)$ and $\Phi_i(\mathfrak{g})=\Phi_i(\mathfrak{g}+\langle  kl\rangle)$.
\item $d_{ij}(\mathfrak{g})=d_{ij}(\mathfrak{g}+\langle  kl\rangle)$ and $\Phi_i(\mathfrak{g})<\Phi_i(\mathfrak{g}+\langle  kl\rangle)$.
\item $d_{ij}(\mathfrak{g})>d_{ij}(\mathfrak{g}+\langle  kl\rangle)$ and $\Phi_i(\mathfrak{g})<\Phi_i(\mathfrak{g}+\langle  kl\rangle)$.
\end{enumerate}
\subsection{Utility Structure}

The utility of agent $i$ in $\mathfrak{g}$ is given by a function $u_i : \mathcal{G} \rightarrow \mathbb{R}^{+}$, where $\mathcal{G}$ is the set of all possible networks on $n$ agents. $u: \mathcal{G} \rightarrow \mathbb{R}^n$ gives the the vector (profile) of utility functions $u = (u_1, . . . ,u_n)$. In other words, each possible resource sharing network structure ($\mathfrak{g}\subseteq \mathcal{G}$) generates a payoff for each agent.

In $\mathfrak{g}$, an agent $i$ gains benefit $\theta_{i}$ and $\xi_{i}$ by accomplishing a computational task and by providing their resource to others, respectively. An agent $i\in \mathfrak{g}$ gains $\xi_i$ with probability $p(1-q)$. An agent $i$'s expected benefit $\theta_{i}$ depends on whether the agent has its own resource or depends on the resource availability in $\mathfrak{g}$. Note that, with probability $pq$, the agent is self-reliant and does not depend on the other agents in $\mathfrak{g}$. However, with probability $q(1-p)$, the agent wants the resource but does not have it and, hence, seeks the resource from other agents in $\mathfrak{g}$. 

An agent seeks resources in $\mathfrak{g}$ by maintaining direct links. Each agent $i$ pays cost $\varsigma_i$ for each of their direct links in $\mathfrak{g}$. Thus, $i$ incurs a total cost of $\eta_{i}(\mathfrak{g})$ $\times$ $\varsigma_i$ in $\mathfrak{g}$. The cost $\varsigma_i$ can be interpreted as the effort or time that agent $i$ spends to maintain an active connection (or link). We consider that agents $i$ and $j$ share the link addition cost equally, that is, $\varsigma=\frac{\varsigma_i+\varsigma_j}{2}$.

Then, for a given resource sharing network $\mathfrak{g}$, the expected payoff of agent $i$ is 
\begin{equation}\label{eq:main-utility}
u_{i}(\mathfrak{g})=p(1-q)\xi_i+q[p+(1-p)\gamma_{i}(\mathfrak{g})]\theta_i-\varsigma \eta_{i}(\mathfrak{g}).
\end{equation}

\section{Resource Availability}
\subsection{Local Resource Availability}
In this section, we discuss our results on the local resource availability of agent $i$ from another agent $j$ in the network $\mathfrak{g}$, that is, 
the probability that agent $i$ will obtain the resource from agent $j$ in $\mathfrak{g}$.

\subsubsection{Link Alteration and Local Resource Availability}\label{subsubsec:link-alteration-local-resource-availability}
In this section, we discuss the effect of link addition and deletion on the local resource availability. 
\begin{lemma}\label{prelemma:link-addition-local-resource-availability}
Suppose $\mathfrak{g}_1$ and $\mathfrak{g}_2$ are resource sharing networks, and suppose $i, j \in \mathfrak{g}_1 \cap \mathfrak{g}_2$. Then $\alpha_{ij}(\mathfrak{g}_1)>\alpha_{ij}(\mathfrak{g}_2)$ if and only if $d_{ij}(\mathfrak{g}_2)\sum\limits_{k\in\mathfrak{g}_2\setminus \{i, j\}} \frac{1}{d_{jk}(\mathfrak{g}_2)}>d_{ij}(\mathfrak{g}_1)\sum\limits_{k\in\mathfrak{g}_1\setminus \{i, j\}} \frac{1}{d_{jk}(\mathfrak{g}_1)}$. 
\end{lemma}
\begin{proof}
$\alpha_{ij}(\mathfrak{g}_1)>\alpha_{ij}(\mathfrak{g}_2)$, if and only if\\

$\frac{p(1-q)}{d_{ij}(\mathfrak{g}_1)(\frac{1}{d_{ij}(\mathfrak{g}_1)}+\sum\limits_{k\in\mathfrak{g}_1\setminus \{i, j\}} \frac{1}{d_{jk}(\mathfrak{g}_1)})}> \frac{p(1-q)}{d_{ij}(\mathfrak{g}_2)(\frac{1}{d_{ij}(\mathfrak{g}_2)}+\sum\limits_{k\in\mathfrak{g}_2\setminus \{i, j\}} \frac{1}{d_{jk}(\mathfrak{g}_2)})}$, \\ 

if and only if\\

$d_{ij}(\mathfrak{g}_2)(\frac{1}{d_{ij}(\mathfrak{g}_2)}+\sum\limits_{k\in\mathfrak{g}_2\setminus \{i, j\}} \frac{1}{d_{jk}(\mathfrak{g}_2)})$\\

$\quad\quad\quad\quad\quad\quad>d_{ij}(\mathfrak{g}_1)(\frac{1}{d_{ij}(\mathfrak{g}_1)}+\sum\limits_{k\in\mathfrak{g}_1\setminus \{i, j\}} \frac{1}{d_{jk}(\mathfrak{g}_1)})$, \\

if and only if\\

$d_{ij}(\mathfrak{g}_2)\sum\limits_{k\in\mathfrak{g}_2\setminus \{i, j\}} \frac{1}{d_{jk}(\mathfrak{g}_2)}>d_{ij}(\mathfrak{g}_1)\sum\limits_{k\in\mathfrak{g}_1\setminus \{i, j\}} \frac{1}{d_{jk}(\mathfrak{g}_1)}.$
\end{proof}

\begin{proposition}\label{lemma:link-addition-local-resource-availability}
Suppose $i$ and $j$ are distinct agents in $\mathfrak{g}$ such that $\langle ij\rangle \not\in \mathfrak{g}$. Then, $\alpha_{ij}(\mathfrak{g}+\langle ij \rangle)>\alpha_{ij}(\mathfrak{g})$.
\end{proposition}
\begin{proof}
Owing to Lemma \ref{prelemma:link-addition-local-resource-availability}, it suffices to show that 
\begin{equation}\label{eq:addition-inequality}
d_{ij}(\mathfrak{g})\sum\limits_{k\in\mathfrak{g}\setminus \{i, j\}} \frac{1}{d_{jk}(\mathfrak{g})}>d_{ij}(\mathfrak{g}+\langle ij \rangle)\sum\limits_{k\in\mathfrak{g}+\langle ij \rangle \setminus \{i, j\}} \frac{1}{d_{jk}(\mathfrak{g}+\langle ij \rangle)}
\end{equation}
Note that $d_{ij}(\mathfrak{g}+\langle ij \rangle) = 1$ and 
$d_{ij}(\mathfrak{g})\in \{2, 3, \ldots, \}$.\\

It suffices to check that Inequality (\ref{eq:addition-inequality}) holds in the following three cases.
\begin{enumerate}
\item $d_{ij}(\mathfrak{g})=\infty$. That is, $i$ and $j$ are not connected in $\mathfrak{g}$. Inequality (\ref{eq:addition-inequality}) clearly holds.
\item $\sum\limits_{k\in\mathfrak{g}\setminus \{i, j\}} \frac{1}{d_{jk}(\mathfrak{g})}=\sum\limits_{k\in\mathfrak{g}+\langle ij \rangle\setminus \{i, j\}} \frac{1}{d_{jk}(\mathfrak{g}+\langle ij \rangle)}$. That is, addition of link $\langle ij \rangle$ does not change the distance between $j$ and any other agent $k$, except $i$. It is easy to see that Inequality (\ref{eq:addition-inequality}) holds in this case too.

\item $\sum\limits_{k\in\mathfrak{g}\setminus \{i, j\}} \frac{1}{d_{jk}(\mathfrak{g})}<\sum\limits_{k\in\mathfrak{g}+\langle ij \rangle\setminus \{i, j\}} \frac{1}{d_{jk}(\mathfrak{g}+\langle ij \rangle)}$. 
This happens when the addition of link $\langle ij \rangle$ changes the distance between $j$ and at least one $k$ (besides $i$). Shortest paths between $j$ and every such $k$ in $\mathfrak{g}+\langle ij \rangle$ are, clearly, shorter than those in $\mathfrak{g}$.
To show that the Left Hand Side of Inequality (\ref{eq:addition-inequality}) is greater than the Right Hand Side, we show that it holds for the worst possible case of the above inequality (given $n$). 
This happens when $d_{ij}(\mathfrak{g})$ is the minimum possible, that is, $2$, and $\mathfrak{g}$, with $n$ agents, is as shown in Figure \ref{fig:local-resource-availability-g-network}. Agent $i$ is connected with $n-2$ agents, agent $j$ has a single neighbour $k$, and $k$ is an intermediary between $i$ and $j$. In $\mathfrak{g}$, $\Phi_{j}(\mathfrak{g})=\frac{2n+3}{6}$ and $d_{ij}(\mathfrak{g})=2$. Suppose agents $i$ and $j$ add a direct link, resulting in the network $\mathfrak{g}+\langle ij\rangle$ as shown in Figure \ref{fig:local-resource-availability-g+ij-network}. Here,  $\Phi_{j}(\mathfrak{g}+\langle ij\rangle)=\frac{n+1}{2}$. 
\begin{figure*}[h!]
\centering
\subfigure[Network $\mathfrak{g}$ ]
{
\includegraphics[scale=0.2]{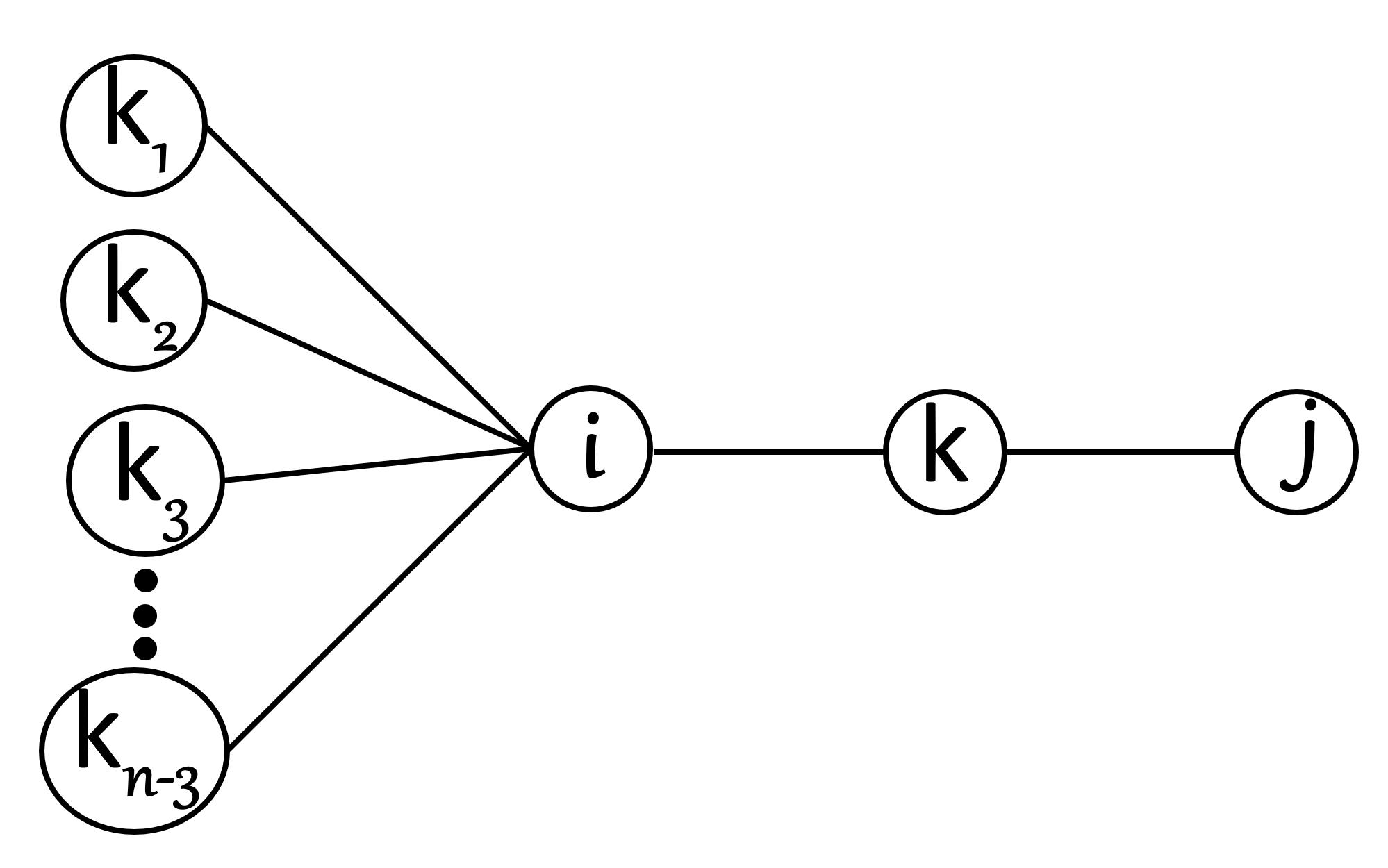} 
\label{fig:local-resource-availability-g-network}
}
\quad 
\subfigure[Network $\mathfrak{g}+\langle ij\rangle$]
{
\includegraphics[scale=0.2]{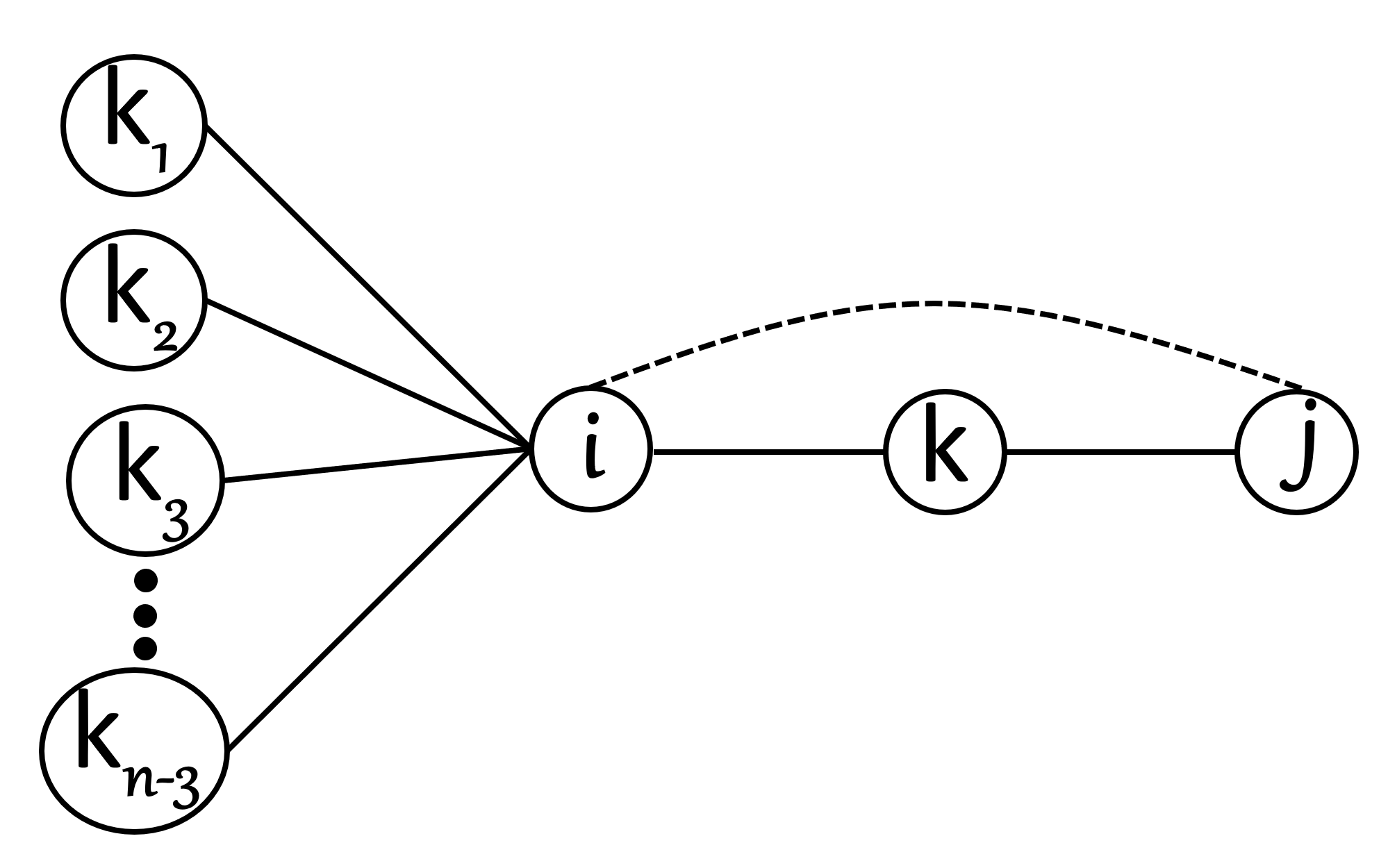} 
\label{fig:local-resource-availability-g+ij-network}
}
\caption{Link addition and local resource availability}
\label{fig:local-resource-availability-addition}
\end{figure*}
%
We have, $\sum\limits_{k\in\mathfrak{g}\setminus \{i, j\}} \frac{1}{d_{jk}(\mathfrak{g})}=\frac{n}{3}$ and $\sum\limits_{k\in\mathfrak{g}+\langle ij \rangle\setminus \{i, j\}} \frac{1}{d_{jk}(\mathfrak{g}+\langle ij \rangle)}=\frac{n-1}{2}$. The Left Hand Side of Inequality (\ref{eq:addition-inequality}) = $2\times\frac{n}{3}$, which is clearly greater than its Right Hand Side, $\frac{n-1}{2}$.
\end{enumerate}
\end{proof}
The above result shows that agents always improve their local resource availabilities by adding new resource sharing connections.
We, now, show that an agent's decision to delete an existing resource sharing connection decreases the local resource availability of the pair of agents from each other.

\begin{proposition}\label{lemma:link-deletion-local-resource-availability}
Suppose $i$ and $j$ are distinct agents in $\mathfrak{g}$ such that $\langle ij\rangle \in \mathfrak{g}$. 
Then, $\alpha_{ij}(\mathfrak{g})>\alpha_{ij}(\mathfrak{g}-\langle ij \rangle)$.
\end{proposition}
\begin{proof}
Owing to Lemma \ref{prelemma:link-addition-local-resource-availability}, it suffices to show that 
{\small
\begin{equation}\label{eq:deletion-inequality}
d_{ij}(\mathfrak{g}-\langle ij \rangle)\sum\limits_{k\in\mathfrak{g}-\langle ij \rangle\setminus \{i, j\}} \frac{1}{d_{jk}(\mathfrak{g}-\langle ij \rangle)} >d_{ij}(\mathfrak{g})\sum\limits_{k\in\mathfrak{g}\setminus \{i, j\}} \frac{1}{d_{jk}(\mathfrak{g})}
\end{equation}
}
We know that $d_{ij}(\mathfrak{g})=1$, $d_{ij}(\mathfrak{g}-\langle ij \rangle) \in \{2, 3, \ldots\}$, and $0 \leq \sum\limits_{k\in\mathfrak{g}\setminus \{i, j\}} \frac{1}{d_{jk}(\mathfrak{g})} \leq n-2$, 0 when $j$ is isolated and $n-2$ when $j$ is connected to all $k$.\\

It suffices to check that Inequality (\ref{eq:deletion-inequality}) holds in the following three cases.

\begin{enumerate}
\item $d_{ij}(\mathfrak{g}-\langle ij \rangle)=\infty$. That is, $\langle ij \rangle$ is the only path between $i$ and $j$ in $\mathfrak{g}$. 
Inequality (\ref{eq:deletion-inequality}), clearly, holds in this case.

\item $\sum\limits_{k\in\mathfrak{g}\setminus \{i, j\}} \frac{1}{d_{jk}(\mathfrak{g})}=\sum\limits_{k\in\mathfrak{g}-\langle ij \rangle\setminus \{i, j\}} \frac{1}{d_{jk}(\mathfrak{g}-\langle ij \rangle)}$.That is, deletion of link $\langle ij \rangle$ does not change the distance between $j$ and any other agent $k$ except $i$. It is easy to see that Inequality (\ref{eq:deletion-inequality}) holds in this case too.

\item $\sum\limits_{k\in\mathfrak{g}\setminus \{i, j\}} \frac{1}{d_{jk}(\mathfrak{g})}>\sum\limits_{k\in\mathfrak{g}-\langle ij \rangle\setminus \{i, j\}} \frac{1}{d_{jk}(\mathfrak{g}-\langle ij \rangle)}$.

To show that the Left Hand Side of Inequality (\ref{eq:deletion-inequality}) is greater than the Right Hand Side, we show that it holds for the worst possible case of the above inequality (given $n$). 
Refer Figure \ref{fig:delete-example-local-resource-availability-g-network}. Here, $\sum\limits_{k\in\mathfrak{g}\setminus \{i, j\}} \frac{1}{d_{jk}(\mathfrak{g})} = \frac{n-1}{2}$.  
\begin{figure*}[h!]
\centering
\subfigure[Network $\mathfrak{g}$ ]
{
\includegraphics[scale=0.20]{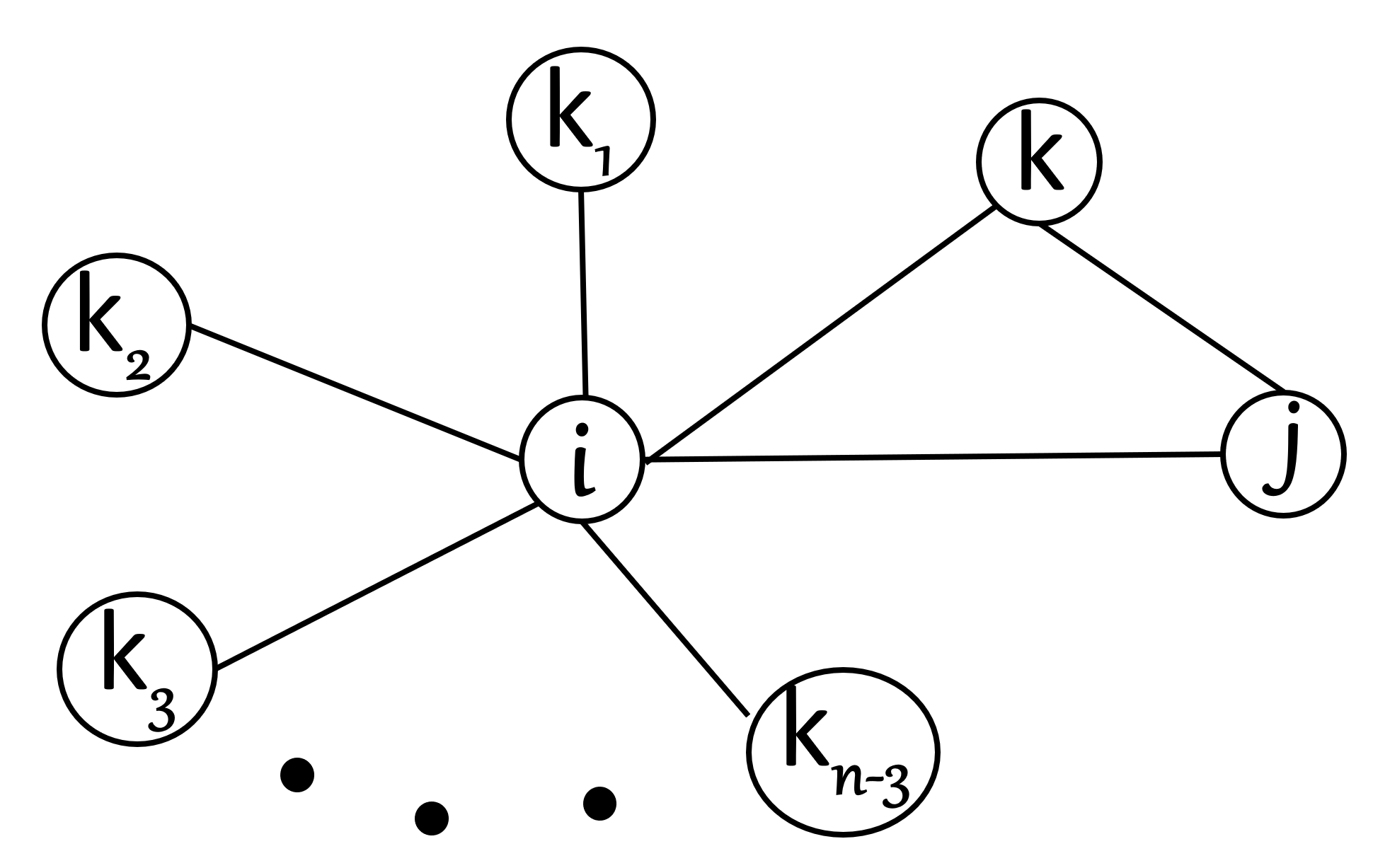} 
\label{fig:delete-example-local-resource-availability-g-network}
}
\quad \quad \quad
\subfigure[Network $\mathfrak{g}-\langle ij\rangle$]
{
\includegraphics[scale=0.20]{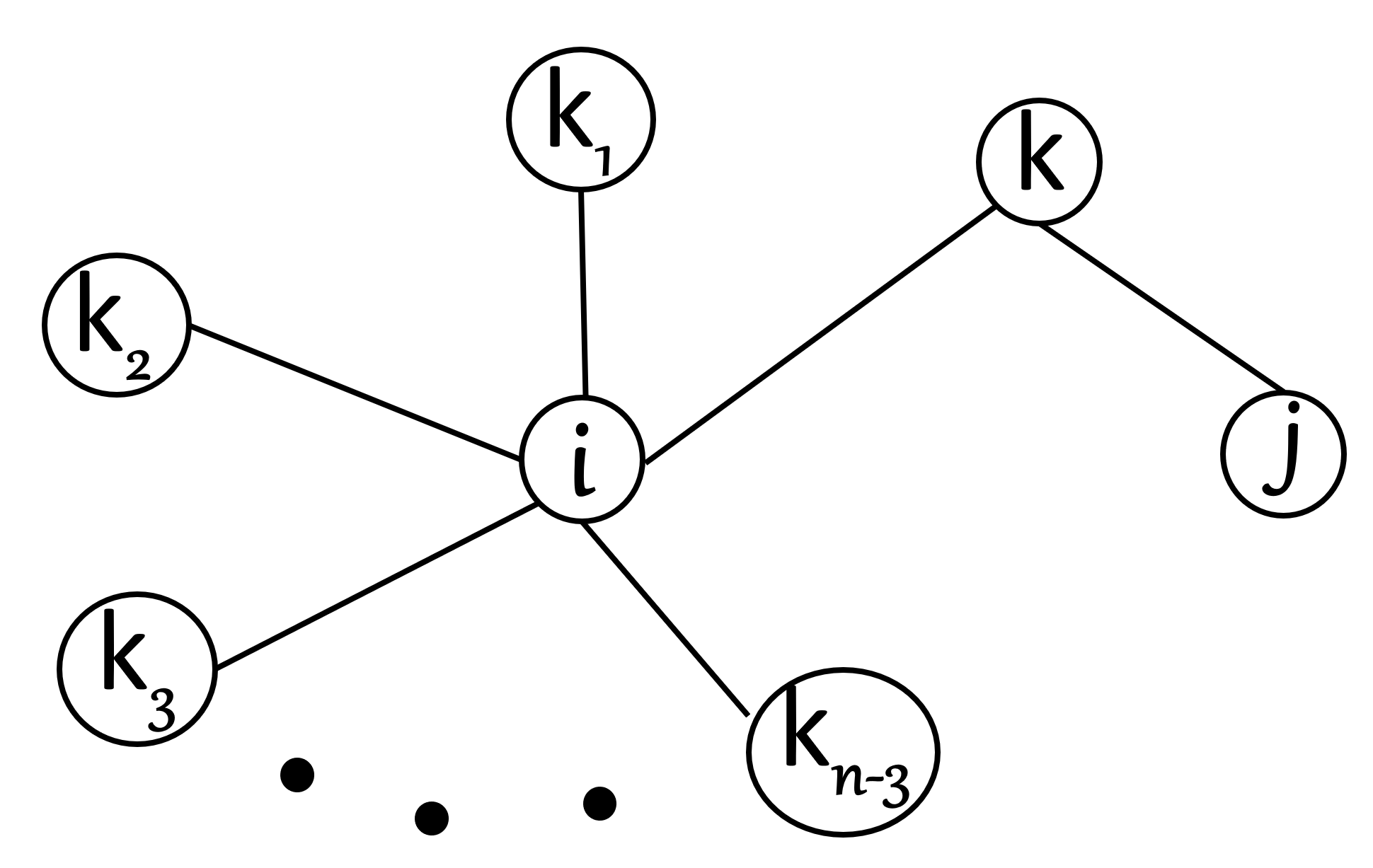} 
\label{fig:local-resource-availability-g-ij-network}
}
\caption{Link deletion and local resource availability}
\label{fig:delete-example-local-resource-availability-addition}
\end{figure*}
%
In $\mathfrak{g}-\langle ij \rangle$, $i$ has links with all agents except $j$. Agent $j$ has a link to at least one other agent, for otherwise we have Case 1. $d_{ij}(\mathfrak{g}-\langle ij \rangle)=2$, the greatest lower bound on the distance, as  $\langle ij \rangle$ has been deleted. Refer Figure \ref{fig:local-resource-availability-g-ij-network}. Here, $\sum\limits_{k\in\mathfrak{g}-\langle ij \rangle\setminus \{i, j\}} \frac{1}{d_{jk}(\mathfrak{g}-\langle ij \rangle)}=\frac{n}{3}$.\\

Inequality (\ref{eq:deletion-inequality}) holds in this case too.\qedhere
\end{enumerate}
\end{proof}

\begin{theorem}
Suppose $i$ and $j$ are distinct agents in $\mathfrak{g}$. Then, the link $\langle ij\rangle$ is always strictly beneficial to both $i$ and $j$, with respect to local resource availabilities from each other. 
\end{theorem}
\begin{proof}
Follows from Propositions \ref{lemma:link-addition-local-resource-availability} and  \ref{lemma:link-deletion-local-resource-availability}.
\end{proof}

\subsubsection{Neighbors and Local Resource Availability}\label{subsubsec:neighbor-local-resource-availability}
In the previous section, we saw that an agent $i$ improves its local resource availability from another agent $j$ by adding link $\langle ij\rangle$. However, this newly added link decreases agent $i$'s local resource availability from its existing neighbors $k$ who are at least three hops away from $j$. For neighbors $k$ who are less than three hops away from $j$, agent $i$'s local resource availability from them remains the same. Similarly, while an agent's local resource availability from another agent decreases if their existing link is deleted, the agent's local resource availability from its existing neighbors who are at least three hops away increases, and remains the same for the other neighbors. We prove these results below. 
\begin{proposition}\label{lemma:local-connection-add-effect}
Suppose $i, j$ and $k$ are distinct agents in $\mathfrak{g}$ such that $\langle ij\rangle \not \in \mathfrak{g}$ and $k\in \eta_i(\mathfrak{g})$. Then, the following hold: 
\begin{enumerate}
    \item If $k\in \eta_j(\mathfrak{g})$, then  $\alpha_{ik}(\mathfrak{g})=\alpha_{ik}(\mathfrak{g}+\langle ij\rangle)$.
    \item If $d_{kj}(\mathfrak{g}) = 2$, then  $\alpha_{ik}(\mathfrak{g})=\alpha_{ik}(\mathfrak{g}+\langle ij\rangle)$.
    \item If $d_{kj}(\mathfrak{g}) > 2$, then $\alpha_{ik}(\mathfrak{g})>\alpha_{ik}(\mathfrak{g}+\langle ij\rangle)$.
\end{enumerate}
\end{proposition}
\begin{proof}
Suppose $d_{kj}(\mathfrak{g}) > 2$. Then, 
$\Phi_{k}(\mathfrak{g}+\langle ij\rangle) > \Phi_{k}(\mathfrak{g})$, as $d_{kj}(\mathfrak{g}+\langle ij\rangle) = 2 < d_{kj}(\mathfrak{g})$, the new shortest path between $k$ and $j$, in $\mathfrak{g}+\langle ij\rangle$, being the path with the two links $\langle ki\rangle$ and $\langle ij\rangle$. 
Therefore, from Eq. (\ref{eq:agent-probability}), 

$\alpha_{ik}(\mathfrak{g}) = \frac{p(1-q)}{\Phi_{k}(\mathfrak{g})}>\frac{p(1-q)}{\Phi_{k}(\mathfrak{g}+\langle ij\rangle)} = \alpha_{ik}(\mathfrak{g}+\langle ij\rangle)$, proving $3$.

\noindent If $d_{kj}(\mathfrak{g}) \leq 2$, then $d_{kj}(\mathfrak{g}+\langle ij\rangle) = d_{kj}(\mathfrak{g})$ and, hence, 

$\Phi_{k}(\mathfrak{g}+\langle ij\rangle) = \Phi_{k}(\mathfrak{g})$, implying 
$\alpha_{ik}(\mathfrak{g}) = \alpha_{ik}(\mathfrak{g}+\langle ij\rangle)$, proving $1$ and $2$.
\end{proof}
Similar results hold for agent $k$'s resource availability from agent $i$ too, as stated in the following corollary.

\begin{corollary}\label{lemma:local-connection-add-effect-on-k}
Suppose $i, j$ and $k$ are distinct agents in $\mathfrak{g}$ such that $\langle ij\rangle \not \in \mathfrak{g}$ and $k\in \eta_i(\mathfrak{g})$. Then, the following hold: 
\begin{enumerate}
    \item If $k\in \eta_j(\mathfrak{g})$, then  $\alpha_{ki}(\mathfrak{g})=\alpha_{ki}(\mathfrak{g}+\langle ij\rangle)$.
    \item If $d_{kj}(\mathfrak{g}) = 2$, then  $\alpha_{ki}(\mathfrak{g})=\alpha_{ki}(\mathfrak{g}+\langle ij\rangle)$.
    \item If $d_{kj}(\mathfrak{g}) > 2$, then $\alpha_{ki}(\mathfrak{g})>\alpha_{ki}(\mathfrak{g}+\langle ij\rangle)$.\\
\end{enumerate}
\end{corollary}

We have the following results on the aggregate local resource availability, aggregated over all neighbors of $i$ in $\mathfrak{g}$. \\
\begin{corollary}\label{corl:local-connection-add-net-effect-2diam}
Suppose $\mathfrak{g}$ is a two-diameter network \footnote{{$\mathfrak{g}$ is a two diameter network if $1\leq d_{ij}(\mathfrak{g})\leq 2$, for all $i, j \in \mathfrak{g}$.}} where $i$ and $j$ are distinct agents such that $\langle ij\rangle \not \in \mathfrak{g}$. Then, \\
$\prod\limits_{k\in\eta_{i}(\mathfrak{g})}\alpha_{ik}(\mathfrak{g})=\prod\limits_{k\in\eta_{i}(\mathfrak{g}+\langle ij\rangle)\setminus\{j\}}\alpha_{ik}(\mathfrak{g}+\langle ij\rangle)$.
\end{corollary}

\begin{corollary}\label{corl:local-connection-add-net-effect}
Suppose $i$ and $j$ are distinct agents in $\mathfrak{g}$ such that $\langle ij\rangle \not \in \mathfrak{g}$. 
If there exists at least one agent $k$ in $\mathfrak{g}$, different from $i$, which is more than two hops away from $j$, then \\
$\prod\limits_{k\in\eta_{i}(\mathfrak{g})}\alpha_{ik}(\mathfrak{g})>\prod\limits_{k\in\eta_{i}(\mathfrak{g}+\langle ij\rangle)\setminus\{j\}}\alpha_{ik}(\mathfrak{g}+\langle ij\rangle)$.
\end{corollary}

We, now, see that an agent's local connections increase its aggregate local resource availability, when the agent deletes a link with one of the neighbors who is at least three hops away from the other neighbors. 
\begin{proposition}\label{lemma:local-connection-delete-effect}
Suppose $i, j$ and $k$ are distinct agents in $\mathfrak{g}$ such that $\langle ij\rangle \in \mathfrak{g}$ and $k\in \eta_i(\mathfrak{g})$. Then, the following hold: 
\begin{enumerate}
    \item If $k\in \eta_j(\mathfrak{g})$, then  $\alpha_{ik}(\mathfrak{g})=\alpha_{ik}(\mathfrak{g}-\langle ij\rangle)$.
    \item If $d_{kj}(\mathfrak{g}) = 2$, then  $\alpha_{ik}(\mathfrak{g})=\alpha_{ik}(\mathfrak{g}-\langle ij\rangle)$.
    \item If $d_{kj}(\mathfrak{g}) > 2$, then $\alpha_{ik}(\mathfrak{g})<\alpha_{ik}(\mathfrak{g}-\langle ij\rangle)$.
\end{enumerate}
\end{proposition}
\begin{proof}
The result can be proved in lines similar to the proof of Proposition \ref{lemma:local-connection-add-effect}. 
\end{proof}
\begin{corollary}\label{lemma:local-connection-delete-effect-on-k}
Suppose $i, j$ and $k$ are distinct agents in $\mathfrak{g}$ such that $\langle ij\rangle \not \in \mathfrak{g}$ and $k\in \eta_i(\mathfrak{g})$. Then, the following hold: 
\begin{enumerate}
    \item If $k\in \eta_j(\mathfrak{g})$, then  $\alpha_{ki}(\mathfrak{g})=\alpha_{ki}(\mathfrak{g}-\langle ij\rangle)$.
    \item If $d_{kj}(\mathfrak{g}) = 2$, then  $\alpha_{ki}(\mathfrak{g})=\alpha_{ki}(\mathfrak{g}-\langle ij\rangle)$.
    \item If $d_{kj}(\mathfrak{g}) > 2$, then $\alpha_{ki}(\mathfrak{g})>\alpha_{ki}(\mathfrak{g}-\langle ij\rangle)$.
\end{enumerate}
\end{corollary}
\begin{corollary}\label{corl:local-connection-del-net-effect-2diam}
Suppose $\mathfrak{g}$ is a two-diameter network where $i$ and $j$ are distinct agents. Suppose $\langle ij\rangle \in \mathfrak{g}$. Then,\\ $\prod\limits_{k\in\eta_{i}(\mathfrak{g})\setminus\{j\}}\alpha_{ik}(\mathfrak{g})=\prod\limits_{k\in\eta_{i}(\mathfrak{g}-\langle ij\rangle)}\alpha_{ik}(\mathfrak{g}-\langle ij\rangle)$.
\end{corollary}

\begin{corollary}\label{corl:local-connection-del-net-effect}
Suppose $i$ and $j$ are distinct agents in $\mathfrak{g}$ such that $\langle ij\rangle \in \mathfrak{g}$. 
If there exists at least one agent $k$ in $\mathfrak{g}$, different from $i$, which is more than two hops away from $j$, then \\
$\prod\limits_{k\in\eta_{i}(\mathfrak{g})\setminus\{j\}}\alpha_{ik}(\mathfrak{g})<\prod\limits_{k\in\eta_{i}(\mathfrak{g}-\langle ij\rangle)}\alpha_{ik}(\mathfrak{g}-\langle ij\rangle)$.
\end{corollary}
\begin{theorem}
Suppose $i$, $j$ and $k$ are distinct agents in $\mathfrak{g}$, such that $k \in \eta_{i}(\mathfrak{g})$. Then, the link $\langle ij\rangle$ is always strictly beneficial to $i$ as well as $k$, with respect to local resource availabilities from each other, if and only if $d_{kj} > 2$. 
\end{theorem}
\begin{proof}
Follows from Propositions 
\ref{lemma:local-connection-add-effect} and \ref{lemma:local-connection-delete-effect}. 
\end{proof}

\subsection{Global Resource Availability}\label{subsec:global-resource-availability}
In this section, we discuss our results on the global resource availability of agent $i$ in $\mathfrak{g}$ which is 
the  probability that $i$ obtains the resource from at least one agent in $\mathfrak{g}$. 

\subsubsection{Link Alteration and Global Resource Availability}\label{subsubsec:global-resource-availability}
In this section, we study the effect of link addition and  deletion on the global resource availability of agents who are involved in these actions. Global resource availability of agents not involved in the link addition or deletion, that is, spillover effect, is discussed in the next section.

In Section \ref{subsubsec:link-alteration-local-resource-availability}, we saw that the local resource availabilities of both agents involved in link addition increase, and decrease for both in the case of link deletion. However, this is not true when we look at the global resource availability of the agents involved. To understand this, we consider the following example.

\begin{example}\label{exmp:gra-resource-availability}
Consider the network $\mathfrak{g}$ as shown in Figure \ref{fig:gra-g-network-link-deletion-example}. If agent $i$ decides to delete the existing link with agent $j$, we have the network $\mathfrak{g}-\langle ij\rangle$ as shown in Figure \ref{fig:gra-g-ij-network-link-deletion-example}. The resource availabilities of agents in these networks are tabulated in Table \ref{table:gra-link-deletion-addition-example} under "Link Deletion". This table shows that agent $i$ benefits by deleting an existing link with $j$ as its resource availability increases by ${\textcolor{blue}{0.00058}}$. However, agent $j$\textquotesingle s resource availability decreases by ${\textcolor{red}{0.13142}}$.

To understand that link addition is not beneficial for an agent, we reverse the above situation. That is, we have network $\mathfrak{g}'$ as shown in Figure \ref{fig:gra-g-ij-network-link-deletion-example}. Now, if agent $i$ decides to add a direct link with agent $j$, we have network $\mathfrak{g}'+\langle ij\rangle$, as shown in Figure \ref{fig:gra-g-network-link-deletion-example}, as a result. The resource availabilities of agents in these networks are tabulated in Table \ref{table:gra-link-deletion-addition-example} under "Link Addition". Here, $i$\textquotesingle s resource availability decreases by ${\textcolor{red}{0.00058}}$ on adding a direct link with $j$, whereas $j$\textquotesingle s resource availability increases by ${\textcolor{blue}{0.13142}}$.
\end{example}
\begin{figure*}[ht!]
\centering
\subfigure[Network $\mathfrak{g}=\mathfrak{g}'+\langle ij\rangle$]
{
\includegraphics[scale=0.20]{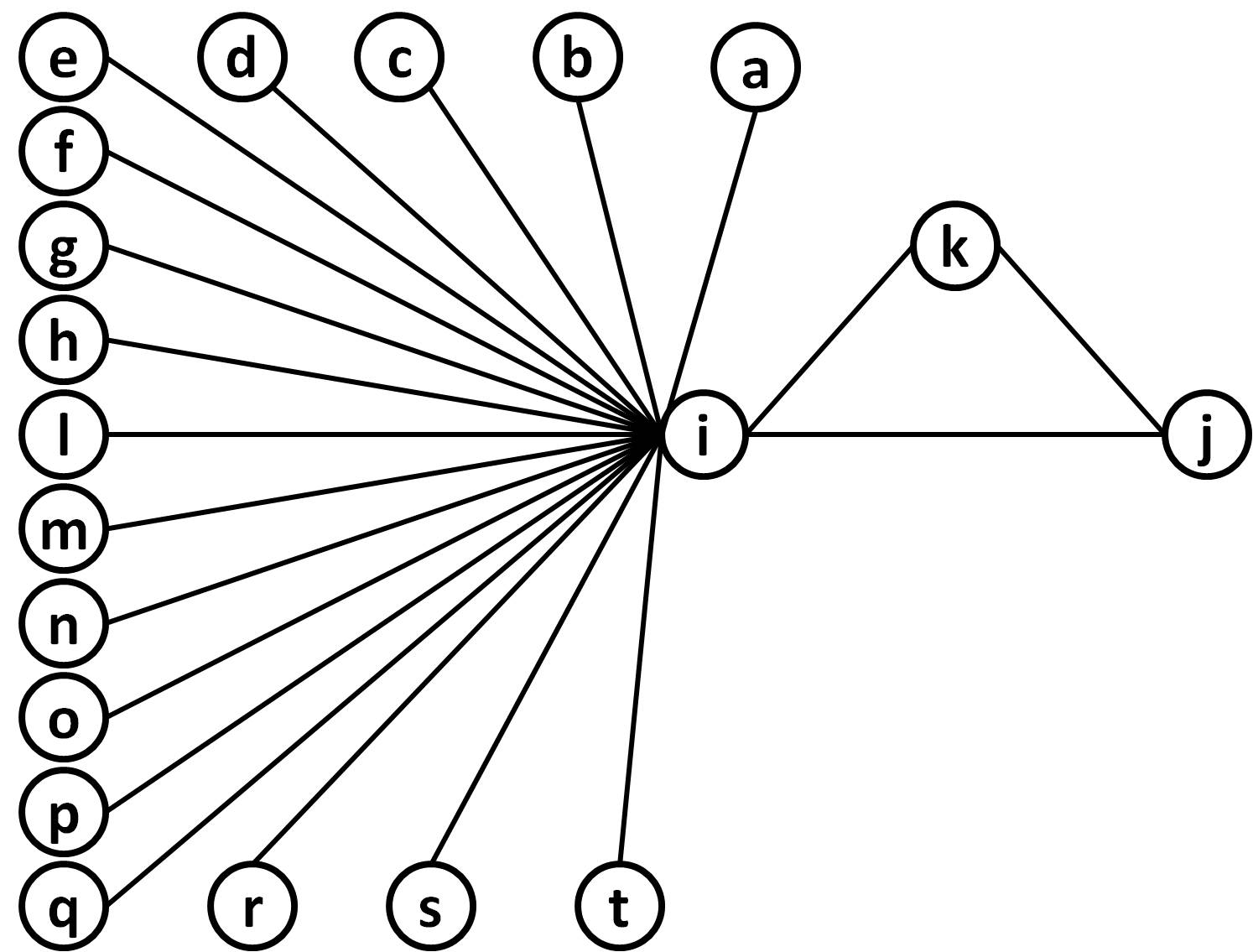} 
\label{fig:gra-g-network-link-deletion-example}
}
\quad \quad \quad
\subfigure[Network $\mathfrak{g}'=\mathfrak{g}-\langle ij\rangle$]
{
\includegraphics[scale=0.20]{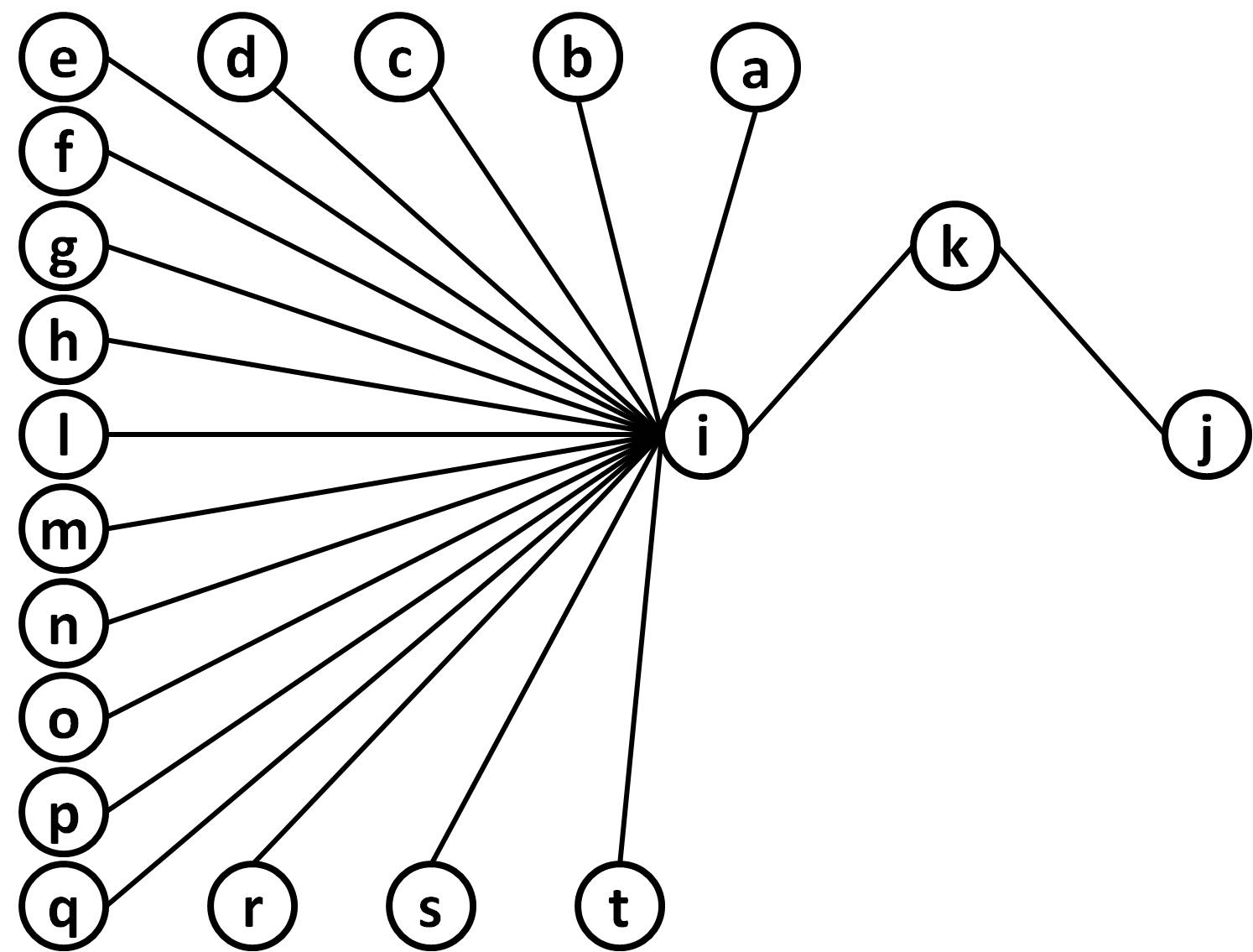} 
\label{fig:gra-g-ij-network-link-deletion-example}
}
\caption{Link addition and deletion}
\label{fig:externalities-first-example1}
\end{figure*}

\begin{table*}[ht!]
\begin{center}
\caption{Link deletion and addition, and global resource availability}
\label{table:gra-link-deletion-addition-example}
\begin{tabular}{|l|l|l|l|l|l|l|}
\hline
 & \multicolumn{3}{c|}{Link Deletion} & \multicolumn{3}{c|}{Link Addition} \\ \hline
\multicolumn{1}{|c|}{Agent} & \multicolumn{1}{c|}{$\gamma_i(\mathfrak{g})$} & \multicolumn{1}{c|}{$\gamma_i(\mathfrak{g}-\langle ij\rangle)$} & $\gamma_i(\mathfrak{g}-\langle ij\rangle)-\gamma_i(\mathfrak{g})$ & \multicolumn{1}{c|}{$\gamma_i(\mathfrak{g}')$} & \multicolumn{1}{c|}{$\gamma_i(\mathfrak{g}'+\langle ij \rangle)$} & \multicolumn{1}{c|}{$\gamma_i(\mathfrak{g}'+\langle ij\rangle)-\gamma_i(\mathfrak{g}')$} \\ \hline
$a$ & 0.62729 & 0.62180 & 0.00549 & 0.62180 & 0.62729 & -0.00549 \\ \hline
$.$ & . & . & . & . & . & . \\ 
$.$ & . & . & . & . & . & . \\ 
$.$ & . & . & . & . & . & . \\ \hline
$t$ & 0.62729 & 0.62180 & 0.00549 & 0.62180 & 0.62729 & -0.00549 \\ \hline
$i$ & 0.86406 & 0.86348 & \textcolor{blue}{0.00058} & 0.86348 & 0.86406 & \textcolor{red}{-0.00058} \\ \hline
$k$ & 0.66479 & 0.64161 & 0.02318 & 0.64161 & 0.66479 & -0.02318 \\ \hline
$j$ & 0.51019 & 0.64161 & \textcolor{red}{-0.13142} & 0.64161 & 0.51019 & \textcolor{blue}{0.13142} \\ \hline
\end{tabular}
\end{center}

\end{table*}

\subsubsection{Spillover Effect}\label{subsec:externalities}

In this section, our focus is to understand the effect that link addition or deletion between a pair of agents has, on the global resource availability of the other agents, that is, the spillover effect. 
\begin{definition}\label{defn:externalities}\citep{jackson-book}
Suppose $i, j$ are distinct agents in $\mathfrak{g}$, such that $\langle ij \rangle \notin \mathfrak{g}$. If agents $i$ and $j$ add the link $\langle ij\rangle$, then, agent $k\in \mathfrak{g}\setminus\{i, j\}$ experiences
				\begin{enumerate}
					\item positive spillover effect due to link addition if \\$u_{k}(\mathfrak{g}+ \langle ij\rangle) > u_{k}(\mathfrak{g})$,
					\item negative spillover effect due to link addition  if \\$u_{k}(\mathfrak{g}+ \langle ij \rangle) < u_{k}(\mathfrak{g})$, and 
					\item no spillover effect due to link addition if \\$u_{k}(\mathfrak{g} +\langle ij \rangle)=u_{k}(\mathfrak{g})$.
				\end{enumerate}
\end{definition}
\cite{Pramod-Games} show that the externalities (spillover effect) that agents experience is determined by their global resource availability. For the sake of completeness, we state the result in the context of resource sharing networks. 
\begin{proposition}\label{prop:externalities} 
Suppose $i, j$ are distinct agents in $\mathfrak{g}$, such that $\langle ij \rangle \notin \mathfrak{g}$. If $i$ and $j$ add the link $\langle ij\rangle$ then, agent $k \in \mathfrak{g} \setminus \{i, j\}$ experiences
				\begin{enumerate}
					\item positive spillover effect due to link addition  if \\ $\gamma_{k}(\mathfrak{g}+ \langle ij\rangle) > \gamma_{k}(\mathfrak{g})$,
					\item negative spillover effect due to link addition  if \\ $\gamma_{k}(\mathfrak{g}+ \langle ij\rangle) < \gamma_{k}(\mathfrak{g})$, and
					\item no spillover effect due to link addition  if \\ $\gamma_{k}(\mathfrak{g} +\langle ij\rangle)=\gamma_{k}(\mathfrak{g})$.
				\end{enumerate}
\end{proposition}
\begin{proof}
Adding $\langle ij\rangle$ does not change the neighbourhood size of $k$. Therefore, from Eq. \ref{eq:main-utility}, $\gamma_{k}(\mathfrak{g})$ increases on adding $\langle ij\rangle$ if and only if $u_{k}(\mathfrak{g})$ increases. 
\end{proof}

Here, we also define the spillover effect when a link is deleted and show a similar result. 
\begin{definition}\label{defn:del-externalities}
Suppose $i, j$ are distinct agents in $\mathfrak{g}$, such that $\langle ij \rangle \in \mathfrak{g}$. If agents $i$ and $j$ delete the link $\langle ij\rangle$ then, agent $k\in \mathfrak{g}\setminus\{i, j\}$ experiences
				\begin{enumerate}
					\item positive spillover effect due to link deletion  if \\$u_{k}(\mathfrak{g}- \langle ij\rangle) > u_{k}(\mathfrak{g})$,
					\item negative spillover effect due to link deletion if \\$u_{k}(\mathfrak{g}- \langle ij \rangle) < u_{k}(\mathfrak{g})$, and 
					\item no spillover effect due to link deletion if \\$u_{k}(\mathfrak{g} +\langle ij \rangle)=u_{k}(\mathfrak{g})$.
				\end{enumerate}
\end{definition}
\begin{proposition}\label{prop:del-externalities}
Suppose $i, j$ are distinct agents in $\mathfrak{g}$, such that $\langle ij \rangle \in \mathfrak{g}$. If agents $i$ and $j$ delete the link $\langle ij\rangle$ then, agent $k \in \mathfrak{g} \setminus \{i, j\}$ experiences
				\begin{enumerate}
					\item positive spillover effect due to link deletion if \\$\gamma_{k}(\mathfrak{g}- \langle ij\rangle) > \gamma_{k}(\mathfrak{g})$,
					\item negative spillover effect due to link deletion if \\$\gamma_{k}(\mathfrak{g}- \langle ij\rangle) < \gamma_{k}(\mathfrak{g})$, and
					\item no spillover effect due to link deletion if \\$\gamma_{k}(\mathfrak{g} -\langle ij\rangle)=\gamma_{k}(\mathfrak{g})$.
				\end{enumerate}
\end{proposition}
\begin{proof}
Similar to the proof of Proposition \ref{prop:externalities}. 
\end{proof}
Now, with the above background, we study how a newly added link affects the chance of an agent obtaining the resource from the other agents.
\begin{lemma}\label{lemma:all-probability}
Suppose $i, j, k, l$ are distinct agents in $\mathfrak{g}$, such that $\langle kl \rangle \not\in \mathfrak{g}$. Then, \begin{enumerate}
\item If $d_{ij}(\mathfrak{g})=d_{ij}(\mathfrak{g}+\langle  kl\rangle)$, the following hold: 
\begin{enumerate}
\item If $\Phi_j(\mathfrak{g})=\Phi_j(\mathfrak{g}+\langle  kl\rangle)$ then \\ $\alpha_{ij}(\mathfrak{g})=\alpha_{ij}(\mathfrak{g}+\langle kl\rangle)$.
\item If $\Phi_j(\mathfrak{g})<\Phi_j(\mathfrak{g}+\langle  kl\rangle)$ then \\ $\alpha_{ij}(\mathfrak{g})>\alpha_{ij}(\mathfrak{g}+\langle kl\rangle)$.
\end{enumerate}
\item If $d_{ij}(\mathfrak{g})>d_{ij}(\mathfrak{g}+\langle  kl\rangle)$, the following hold: 
\begin{enumerate}
\item If $\frac{d_{ij}(\mathfrak{g}+\langle  kl\rangle)}{d_{ij}(\mathfrak{g})}>\frac{\Phi_j(\mathfrak{g})}{\Phi_j(\mathfrak{g}+\langle  kl\rangle)}$ then \\ $\alpha_{ij}(\mathfrak{g})>\alpha_{ij}(\mathfrak{g}+\langle kl\rangle)$.
\item If $\frac{d_{ij}(\mathfrak{g}+\langle  kl\rangle)}{d_{ij}(\mathfrak{g})}<\frac{\Phi_j(\mathfrak{g})}{\Phi_j(\mathfrak{g}+\langle  kl\rangle)}$ then \\ $\alpha_{ij}(\mathfrak{g})<\alpha_{ij}(\mathfrak{g}+\langle kl\rangle)$.
\end{enumerate}
\end{enumerate}
\end{lemma}

\begin{proof}
The proof follows as\\
$\alpha_{ij}(\mathfrak{g})=\frac{1}{d_{ij}(\mathfrak{g})\Phi_j(\mathfrak{g})}$ and $\alpha_{ij}(\mathfrak{g}+\langle  kl\rangle)=\frac{1}{d_{ij}(\mathfrak{g}+\langle  kl\rangle)\Phi_j(\mathfrak{g}+\langle  kl\rangle)}$.
\end{proof}
{Further, in this paper, we study the role of closeness in determining spillover effects of agents, and show that agents always experience either positive or negative spillover effect, and the case of "no spillover effect" never happens.} \\
\begin{proposition}\label{prop:negative-externalities}
Suppose $i, j$ are distinct agents in $\mathfrak{g}$, such that $\langle ij \rangle \notin \mathfrak{g}$. An agent $k \in \mathfrak{g}\setminus\{i, j\}$ experiences negative spillover effect on addition of the link $\langle ij\rangle$, if $\Phi_{k}(\mathfrak{g})=\Phi_{k}(\mathfrak{g}+\langle ij\rangle)$.
\end{proposition}
\begin{proof}
Let agent $i$ and $j$ add a direct link in $\mathfrak{g}$. This new link $\langle ij\rangle$ reduces their distance by at least $1$, and by at most $d_{ij}(\mathfrak{g})-1$ in $\mathfrak{g}+\langle ij\rangle$. Thus, their closeness increases by at least $\frac{d_{ij}(\mathfrak{g})-1}{d_{ij}(\mathfrak{g})}$ in $\mathfrak{g}+\langle ij\rangle$.\\

\noindent Therefore, $\Phi_{i}(\mathfrak{g})<\Phi_{i}(\mathfrak{g}+\langle ij\rangle)$ and $\Phi_{j}(\mathfrak{g})<\Phi_{j}(\mathfrak{g}+\langle ij\rangle)$.\\

\noindent As $\Phi_{k}(\mathfrak{g})=\Phi_{k}(\mathfrak{g}+\langle ij\rangle)$, we have $d_{kl}(\mathfrak{g})=d_{kl}(\mathfrak{g}+\langle ij\rangle)$, for all $l\in \mathfrak{g}$. In particular, $d_{ik}(\mathfrak{g})=d_{ik}(\mathfrak{g}+\langle ij\rangle)$. \\

\noindent Therefore, $\alpha_{ki}(\mathfrak{g})>\alpha_{ki}(\mathfrak{g}+\langle ij\rangle)$. Similarly, \\$\alpha_{kj}(\mathfrak{g})>\alpha_{kj}(\mathfrak{g}+\langle ij\rangle)$. Hence, \\
\noindent $(1-\alpha_{ki}(\mathfrak{g}))(1-\alpha_{kj}(\mathfrak{g}))<(1-\alpha_{ki}(\mathfrak{g}+\langle ij\rangle))(1-\alpha_{kj}(\mathfrak{g}+\langle ij\rangle))$.\\

\noindent This implies $\gamma_{k}(\mathfrak{g})>\gamma_{k}(\mathfrak{g}+\langle ij\rangle)$, as 
$\alpha_{kl}(\mathfrak{g})\geq \alpha_{kl}(\mathfrak{g}+\langle ij\rangle)$ for all $l\in \mathfrak{g}\setminus\{i, j\}$ too
(since $\Phi_{l}(\mathfrak{g})\leq \Phi_{l}(\mathfrak{g}+\langle ij\rangle)$).\qedhere

\end{proof}
\begin{corollary}\label{exmp:2-diameter} 
Suppose $i, j$ are distinct agents in $\mathfrak{g}$, such that $\langle ij \rangle \notin \mathfrak{g}$.  $\Phi_k(\mathfrak{g})=\Phi_k(\mathfrak{g}+\langle ij\rangle)$ for all $ k\in \mathfrak{g}\setminus\{i, j\}$ and hence, all agents experience only negative spillover effect if $\langle ij\rangle$ is added.
\end{corollary}
\begin{remark}\label{remark:sufficiency}
Proposition \ref{prop:negative-externalities} shows that an increase in agents' closeness is necessary for them to experience positive spillover effect. However, the increase in agents' closeness is not a sufficient condition for positive spillover effect, as demonstrated by
the following example.
\end{remark}
\begin{example}\label{exmp:externalities}
Consider the Pedgetts's Florentine\_Families network \citep{BREIGER1986215} (showing business and marital ties of 16 agents) generated by Social Network Visualizer\footnote{https://socnetv.org/} tool from its known data set, as shown in Figure \ref{fig:externalities-example-network-g}. Call this network  $\mathfrak{g}$. Suppose Medici and Strozzi add a link in $\mathfrak{g}$ resulting in the network $\mathfrak{g}' = \mathfrak{g}+\langle Medici, Strozzi\rangle$, as shown in Figure \ref{fig:externalities-example-network-g+ij}. From Eq. (\ref{eq:harmonic-closeness}), $\Phi_{Albizzi}(\mathfrak{g})=7.83$ and $\Phi_{Albizzi}(\mathfrak{g}')=8.00$. However, from Eq. (\ref{eq:gra}), $\gamma_{Albizzi}(\mathfrak{g})= 0.701$ and $\gamma_{Albizzi}(\mathfrak{g}')= 0.696$. This indicates that, although the closeness of Albizzi increases in $\mathfrak{g}'$, its global resource availability decreases. Similar is true in the case of Ginori and Pazzi. On the contrary, the newly added link between Medici and Strozzi increases not only Bischeri's closeness (from $\Phi_{Bischeri}(\mathfrak{g})=7.20$ to $\Phi_{Bischeri}(\mathfrak{g}')=7.58$), but also its global resource availability (from $ \gamma_{Bischeri}(\mathfrak{g})=0.657$ to $\gamma_{Bischeri}(\mathfrak{g}')= 0.662$). Similar holds for Acciaiuodi, Peruzzi and Salviati too. Note that $\gamma_{Pucci}(\mathfrak{g}) = \gamma_{Pucci}(\mathfrak{g}') = 0$.
\end{example}
\begin{figure*}[h!]
\centering
\subfigure[Network $\mathfrak{g} = \mathfrak{g}'-\langle Medici, Strozzi \rangle$]
{
\includegraphics[scale=0.27]{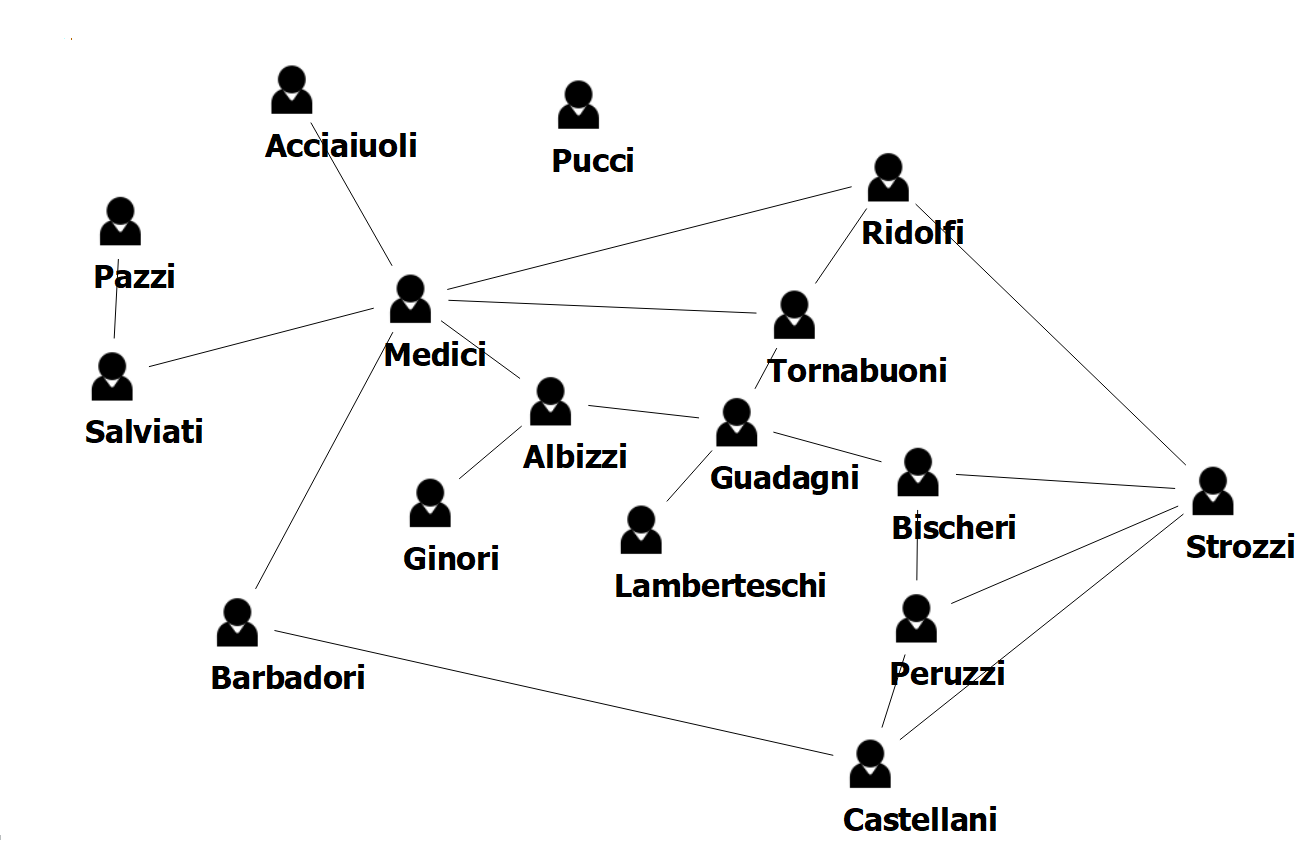} 
\label{fig:externalities-example-network-g}
}
\quad \quad \quad
\subfigure[Network $\mathfrak{g}' = \mathfrak{g}+\langle Medici, Strozzi\rangle$]
{
\includegraphics[scale=0.27]{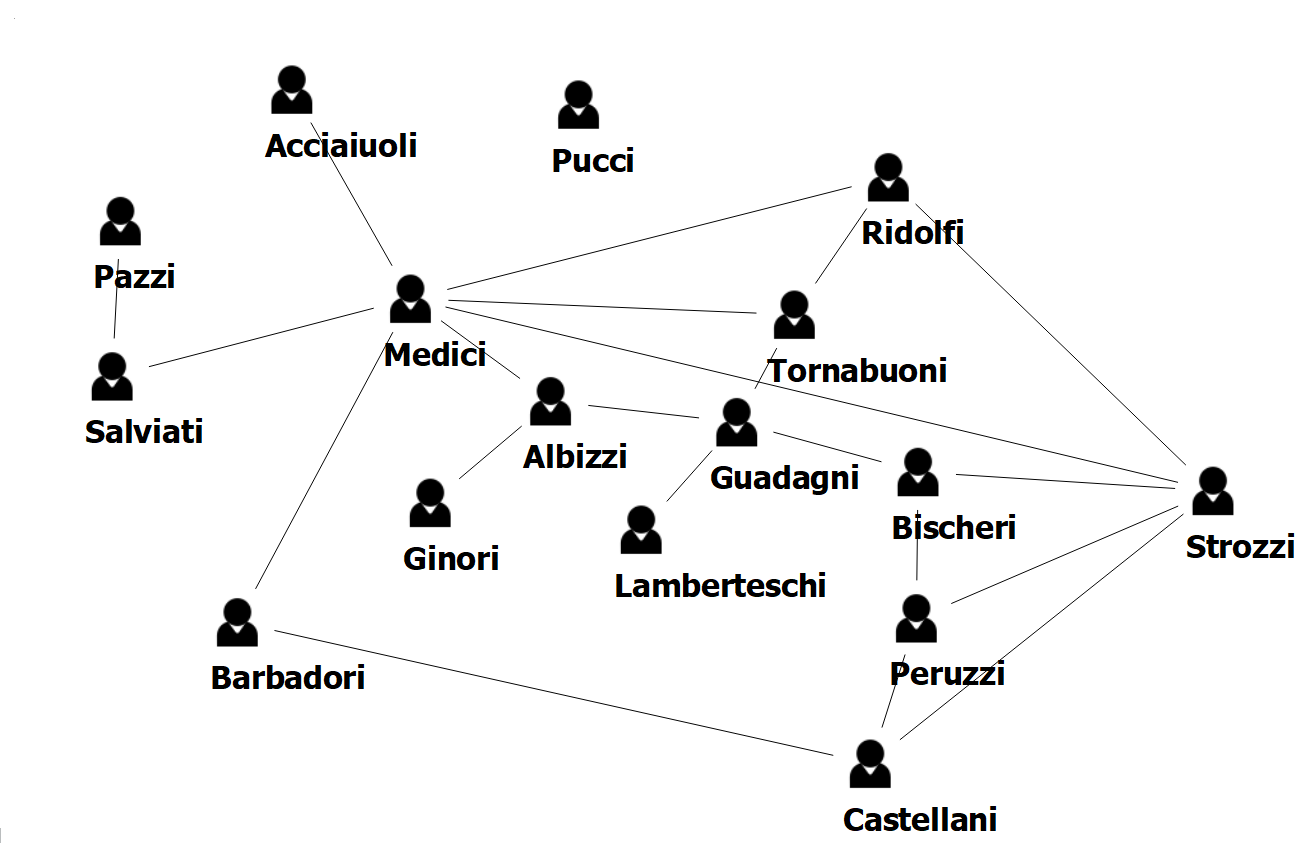} 
\label{fig:externalities-example-network-g+ij}
}
\caption{Spillover in {resource sharing network} $\mathfrak{g}$}
\label{fig:externalities-example}
\end{figure*}
We provide similar results for spillover effects due to link deletion too. 
\begin{proposition}\label{prop:positive-dexternalities}
Suppose $i, j, k$ are distinct agents in $\mathfrak{g}$, such that $\langle ij \rangle \in \mathfrak{g}$. $k$ experiences positive spillover effect due to link deletion, on deletion of $\langle ij\rangle$, if $\Phi_{k}(\mathfrak{g})=\Phi_{k}(\mathfrak{g}-\langle ij\rangle)$.
\end{proposition}
\begin{proof}
The result can be proved in lines similar to the proof of Proposition \ref{prop:negative-externalities}. 
\end{proof}

\begin{example}\label{exmp:externalities}
Consider the resource sharing network, $\mathfrak{g}'$, shown in Figure \ref{fig:externalities-example-network-g+ij}. Let Medici and Strozzi delete the link between them in $\mathfrak{g}'$ resulting in the {resource sharing network} $\mathfrak{g}'-\langle Medici, Strozzi\rangle$ ($\mathfrak{g}$, as shown in Figure \ref{fig:externalities-example-network-g}. From Eq. (\ref{eq:harmonic-closeness}), $\Phi_{Albizzi}(\mathfrak{g}')=8.00$ and $\Phi_{Albizzi}(\mathfrak{g})=7.83$. However, from Eq. (\ref{eq:gra}), $\gamma_{Albizzi}(\mathfrak{g}')= 0.696$ and $\gamma_{Albizzi}(\mathfrak{g})= 0.701$. This indicates that, although the closeness of Albizzi decreases in $\mathfrak{g}'-\langle ij \rangle$, its global resource availability increases. Similar is true in the case of agents Ginori and Pazzi. On the contrary, the link deletion between Medici and Strozzi in $\mathfrak{g}'$ decreases not only agent Bischeri's closeness (from  $\Phi_{Bischeri}(\mathfrak{g}')=7.58$ to $\Phi_{Bischeri}(\mathfrak{g})=7.20$), but also its global resource availability (from $\gamma_{Bischeri}(\mathfrak{g}')= 0.662$ to $\gamma_{Bischeri}(\mathfrak{g}$). We have a similar observation for Acciaiuodi, Peruzzi and Salviati too. $\gamma_{Pucci}(\mathfrak{g}') = \gamma_{Pucci}(\mathfrak{g}) = 0.$
\end{example}
Now, we show that 
in any connected {resource sharing network} with three or more agents, all agents experience either positive or negative spillover effects, 
and there is no case where any agent experiences no spillover effect.

\begin{theorem}\label{remark:sufficiency}
Suppose $\mathfrak{g}$ is connected and has three or more agents. Suppose $k$ is any agent in $\mathfrak{g}$. Agent $k$ always experiences either positive or negative spillover effect, and it is not possible that $k$ experiences no spillover effect.
\end{theorem}
\begin{proof}
We prove the result for spillover effect due to link addition, and a similar proof works for spillover effect due to link deletion too. \\

Suppose $i$, $j$ and $k$ are distinct agents in $\mathfrak{g}$, such that $\langle ij\rangle \notin \mathfrak{g}$. If possible, let $k$ have no spillover effect when $\langle ij\rangle$ is added. This means 
$\gamma_{k}(\mathfrak{g})=\gamma_{k}(\mathfrak{g}+\langle ij\rangle)$, by Proposition \ref{prop:externalities}.\\

Therefore, by Eq. (\ref{eq:gra}),   
$\alpha_{kl}(\mathfrak{g})=\alpha_{kl}(\mathfrak{g}+\langle ij\rangle)$, for all $l\in \mathfrak{g}$. (Note that $\alpha_{kl}(\mathfrak{g})\geq\alpha_{kl}(\mathfrak{g}+\langle ij\rangle)$, for all $l\in \mathfrak{g}$).\\

This implies that $d_{kl}(\mathfrak{g})=d_{kl}(\mathfrak{g}+\langle ij\rangle)$ 
and $\Phi_l(\mathfrak{g})=\Phi_l(\mathfrak{g}+\langle ij\rangle)$ for all $l\in \mathfrak{g}$, from Equations (\ref{eq:harmonic-closeness}) and   (\ref{eq:agent-probability}).\\

The link addition between agents $i$ and $j$ in $\mathfrak{g}$ decreases their distance by $d_{ij}(\mathfrak{g})-1$ in $\mathfrak{g}+\langle ij\rangle$, and thus, the closeness of both $i$ and $j$ in $\mathfrak{g}+\langle ij\rangle$ increases by $\frac{1}{d_{ij}(\mathfrak{g})-1}$.\\

Hence, $\Phi_i(\mathfrak{g})<\Phi_i(\mathfrak{g}+\langle ij\rangle)$ and $\Phi_j(\mathfrak{g})<\Phi_j(\mathfrak{g}+\langle ij\rangle)$, a contradiction to our deduction that $\Phi_l(\mathfrak{g})=\Phi_l(\mathfrak{g}+\langle ij\rangle)$ for all $l\in \mathfrak{g}$.\\

Therefore, our assumption that $k$ has no spillover effect is incorrect. In other words, $k$ has either positive or negative spillover effect. 
\end{proof}
\begin{corollary}
Suppose $\mathfrak{g}$ is disconnected with at least three disjoint components. Suppose $i, j$, and $k$ are three distinct agents in $\mathfrak{g}$ such that $i\in \mathfrak{g}(\mathfrak{c}_x)$, $j\in \mathfrak{g}(\mathfrak{c}_y)$, and $k\in \mathfrak{g}(\mathfrak{c}_z)$, where $\mathfrak{c}_x, \mathfrak{c}_y$ and $\mathfrak{c}_z$ are disjoint. Suppose agents $i$ and $j$ add a direct link in $\mathfrak{g}$, then $\gamma_{k}(\mathfrak{g})=\gamma_{k}(\mathfrak{g}+\langle ij\rangle)$ for all $k \not \in \mathfrak{g}(\mathfrak{c}_x), \mathfrak{g}(\mathfrak{c}_y)$.
\end{corollary}
\section{Choice Modelling}
To understand the relation between agents' distance from each other, and local as well as global resource availabilities, we, first, discuss the following example. 

\begin{example}

Suppose $\mathfrak{g}$ is a ring network\footnote{A ring network $\mathfrak{g}$ is a connected network where $\eta_i(\mathfrak{g})=2$ for all $i \in \mathfrak{g}$.}. 
Then, for all $i\in \mathfrak{g}$, 
\[
\Phi_{i}(\mathfrak{g})=
\begin{cases}
    2(1+\frac{1}{2}+\frac{1}{3}+ \ldots +\frac{1}{N}), ~\text{if N is odd} \\
    2(1+\frac{1}{2}+\frac{1}{3}+ \ldots +\frac{1}{N-1})+\frac{1}{N}, ~\text{if N is even.} 
\end{cases}
\]

Suppose $i \in \mathfrak{g}$. For $j \in \mathfrak{g}$, suppose $\langle ij\rangle \notin \mathfrak{g}$. We compute the local resource availability, $\alpha_{ij}(\mathfrak{g}+\langle ij\rangle)$ of $i$ from $j$, and the global resource availability $\gamma_{i}(\mathfrak{g}+\langle ij\rangle)$ for different $j$ in increasing order of the distance between $i$ and $j$ in $\mathfrak{g}$. Figures \ref{fig:local-resource-availability} and \ref{fig:global-resource-availability} show that as the distance between $i$ and $j$ increases, $\alpha_{ij}(\mathfrak{g}+\langle ij\rangle)$ and $\gamma_{i}(\mathfrak{g}+\langle ij\rangle)$ also increase.
\end{example}

\begin{figure*}[h!]
\centering
\subfigure[$\alpha_{ij}(\mathfrak{g}+\langle ij\rangle)$]
{
\includegraphics[scale=0.23]{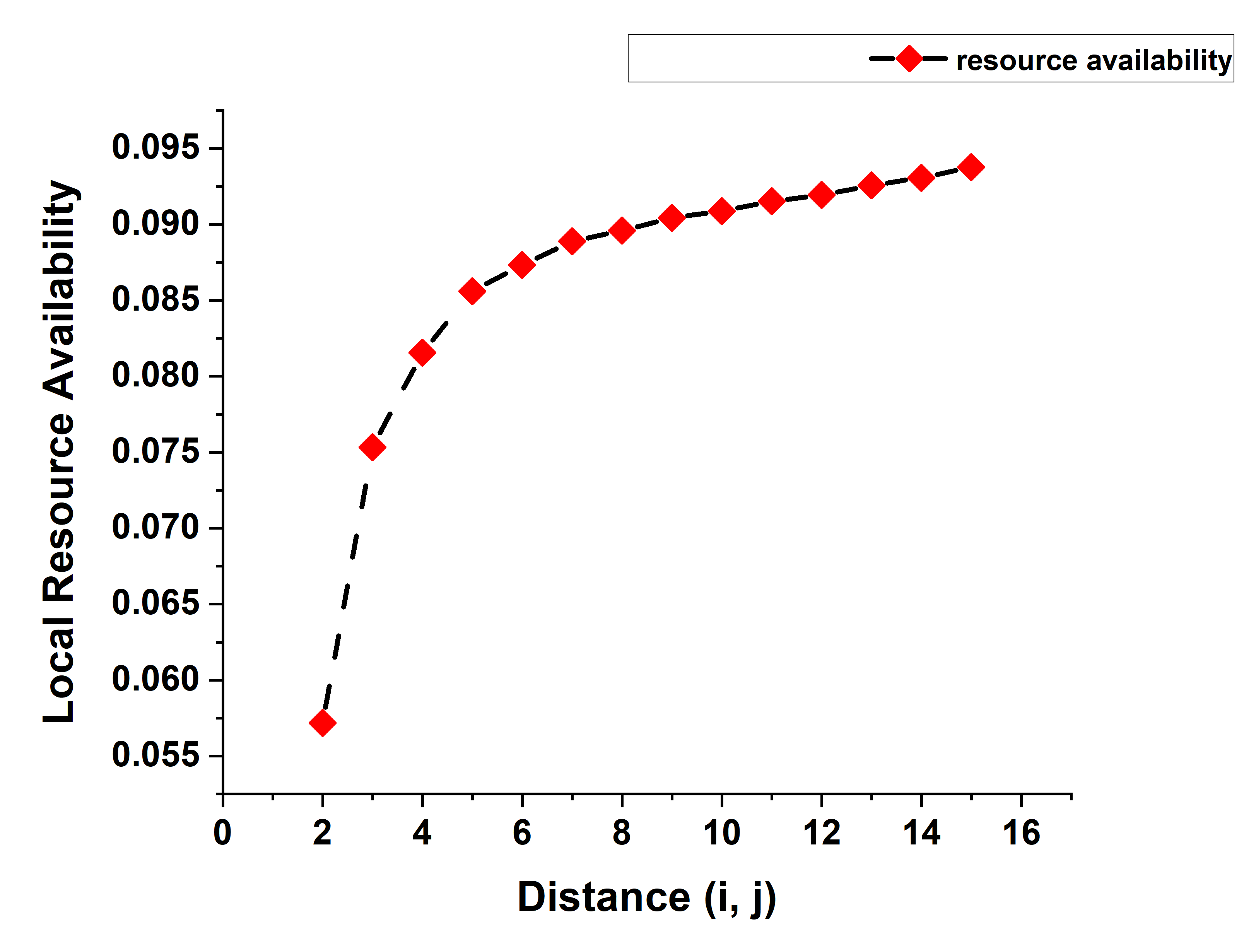} 
\label{fig:local-resource-availability}
}
\subfigure[$\gamma_{i}(\mathfrak{g}+\langle ij\rangle)$]
{
\includegraphics[scale=0.23]{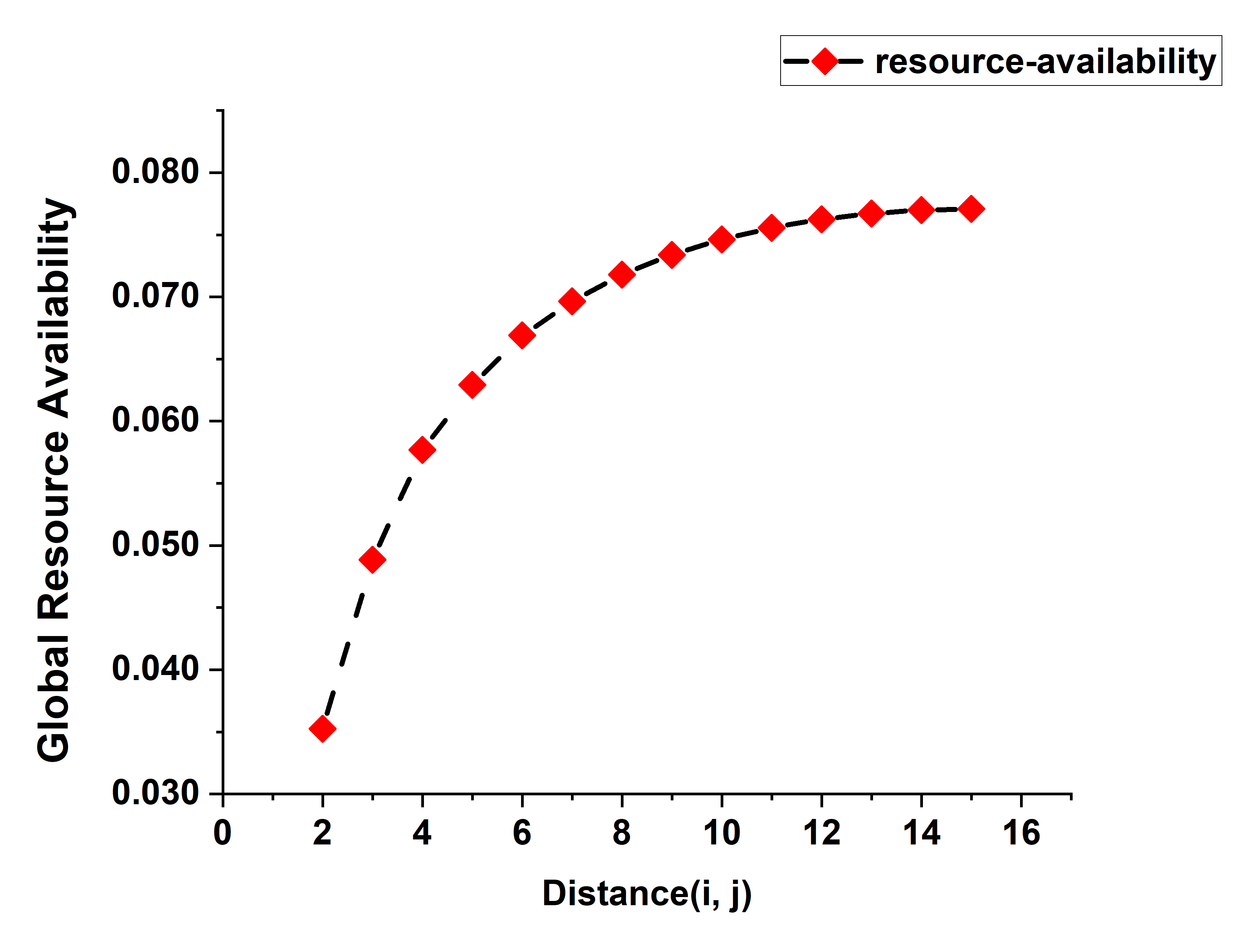} 
\label{fig:global-resource-availability}
}
\caption{Local and global resource availabilities of agent $i$ in the ring network}
\label{fig:local-global-resource-availability}
\end{figure*}
We, now, discuss the relation between local resource availability and distance as well as closeness. 

\begin{lemma}\label{starter-lemma-local-resource-availability}
Suppose $i, j$ are distinct agents in $\mathfrak{g}$, such that $\langle ij \rangle \not\in \mathfrak{g}$. The local resource availability of agent $i$ from agent $j$ increases with decrease in the distance, $d_{ij}(\mathfrak{g})$, between them. 

\end{lemma}
\begin{proof}
For any $i, j \in \mathfrak{g},~i \not= j$, $0 < d_{ij}(\mathfrak{g}) < d_{ij}(\mathfrak{g})+1$. Then from Eq. (\ref{eq:agent-probability}), $\frac{p(1-q)}{d_{ij}(\mathfrak{g}) \Phi_j(\mathfrak{g})}>\frac{p(1-q)}{(d_{ij}(\mathfrak{g})+1) \Phi_j(\mathfrak{g})}$.

\end{proof}

\begin{lemma}\label{starter-lemma2-local-resource-availability}
Suppose $i, j, k$ are distinct agents in $\mathfrak{g}$, such that $\langle ij \rangle, \langle ik \rangle \not\in \mathfrak{g}$. 
If $d_{ij}(\mathfrak{g}) > d_{ik}(\mathfrak{g})$ then, $\alpha_{ij}(\mathfrak{g})<\alpha_{ik}(\mathfrak{g})$. 

\end{lemma}
\begin{proof}
Follows from Lemma \ref{starter-lemma-local-resource-availability}. 
\end{proof}

\begin{lemma}\label{starter-lemma-max-local-resource-availability}
Agent $i$ in $\mathfrak{g}$ 
obtains maximum local resource availability from that $k\in\mathfrak{g}$ who is least close to others (that is, with the least harmonic centrality). 
\end{lemma}
\begin{proof}
Let $j, k \in\mathfrak{g}$ and $0< \Phi_k(\mathfrak{g})<\Phi_j(\mathfrak{g})$, and $d_{ij}(\mathfrak{g})=d_{ik}(\mathfrak{g})$. Then from Eq. (\ref{eq:agent-probability}),\\
$\frac{p(1-q)}{d_{ik}(\mathfrak{g}) \Phi_k(\mathfrak{g})}>\frac{p(1-q)}{d_{ij}(\mathfrak{g}) \Phi_j(\mathfrak{g})}$. \qedhere

\end{proof}
We showed, in Lemmas \ref{starter-lemma2-local-resource-availability} and \ref{starter-lemma-max-local-resource-availability}, that the local resource availability of an agent from another agent increases with decrease in the distance between them and that maximum local resource availability is obtained from the agent with the least closeness (that is, least harmonic centrality). 
We, now, look at the relation between the global resource availability of an agent and its closeness.  

\begin{lemma}
Agent $i$ in $\mathfrak{g}$ maximizes its global resource availability by maximizing its own closeness or equivalently, by minimizing its distance with others.
\end{lemma}
\begin{proof}
The proof is in lines similar to that of Lemma \ref{starter-lemma-max-local-resource-availability}. 
\end{proof}

We now discuss results that show which agent to add a link to, so as to maximize the local resource availability. Such results are difficult to establish for global resource availability because, as seen in Example \ref{exmp:gra-resource-availability}, if agents $i$ and $j$ add or delete the link between them, the global resource availability of one of them may increase while that of the other may decrease. 

\begin{proposition}\label{choice-based-on-local}
Suppose $i$ is an agent in $\mathfrak{g}$. Across all agents $j$ in $\mathfrak{g}$ such that $\langle ij \rangle \notin \mathfrak{g}$, suppose $i$ chooses $j=j_0$ to which to add a link, then $j_0$ maximizes the local resource availability of $i$ from $j$ in $\mathfrak{g}+\langle ij \rangle$ if and only if $j_0$ is the agent (or one of the agents) whose {closeness is the least}. 
\end{proposition}
\begin{proof}
Suppose agents $k$ and $l$ $\in \mathfrak{g}$ are such that $\langle ik \rangle \notin \mathfrak{g}$ and $\langle il \rangle \notin \mathfrak{g}$. 
Agent $i$ prefers $l$ over $k$ to add a link, if and only if \\

$\alpha_{il}(\mathfrak{g}+\langle il\rangle)>\alpha_{ik}(\mathfrak{g}+\langle ik\rangle)$, if and only if\\

$\frac{1}{d_{il}(\mathfrak{g}+\langle il\rangle)\Phi_l(\mathfrak{g}+\langle il\rangle)}>\frac{1}{d_{ik}(\mathfrak{g}+\langle ik\rangle)\Phi_k(\mathfrak{g}+\langle ik\rangle)}$. \\

We have $d_{il}(\mathfrak{g}+\langle il\rangle)=d_{ik}(\mathfrak{g}+\langle ik\rangle)=1$. Hence,\\

$\alpha_{il}(\mathfrak{g}+\langle il\rangle)>\alpha_{ik}(\mathfrak{g}+\langle ik\rangle)$, if and only if  \\

$\frac{1}{\Phi_l(\mathfrak{g}+\langle il\rangle)}>\frac{1}{\Phi_k(\mathfrak{g}+\langle ik\rangle)}$, if and only if\\

$\Phi_l(\mathfrak{g}+\langle il\rangle) < \Phi_k(\mathfrak{g}+\langle ik\rangle)$.

\end{proof}

\begin{lemma}\label{lemma-choice-max-diff-local}
Suppose $i, j, k$ are distinct agents in $\mathfrak{g}$. With respect to the local resource availability, $i$ prefers $j$ over $k$ to which to add a link, if and only if\\

$ \frac{\Phi_k(\mathfrak{g}+\langle ik\rangle)-\Phi_j(\mathfrak{g}+\langle ij\rangle)}{\Phi_k(\mathfrak{g}+\langle ik\rangle)\Phi_j(\mathfrak{g}+\langle ij\rangle)}>\frac{d_{ik}(\mathfrak{g})\Phi_k(\mathfrak{g})-d_{ij}(\mathfrak{g})\Phi_j(\mathfrak{g})}{[d_{ik}(\mathfrak{g})\Phi_k(\mathfrak{g})][d_{ij}(\mathfrak{g})\Phi_j(\mathfrak{g})]}$.
\end{lemma}
\begin{proof}
$\alpha_{ij}(\mathfrak{g}+\langle ij \rangle)-\alpha_{ij}(\mathfrak{g}) > \alpha_{ik}(\mathfrak{g}+\langle ik \rangle)-\alpha_{ik}(\mathfrak{g})$, \\

\noindent if and only if
$\alpha_{ij}(\mathfrak{g}+\langle ij \rangle)-\alpha_{ik}(\mathfrak{g}+\langle ik \rangle) > \alpha_{ij}(\mathfrak{g})-\alpha_{ik}(\mathfrak{g})$, \\

\noindent if and only if
$\frac{1}{\Phi_j(\mathfrak{g}+\langle ij\rangle)}-\frac{1}{\Phi_k(\mathfrak{g}+\langle ik\rangle)} > \frac{1}{d_{ij}(\mathfrak{g})\Phi_j(\mathfrak{g})}-\frac{1}{d_{ik}(\mathfrak{g})\Phi_k(\mathfrak{g})}$, \\

\noindent if and only if
$\frac{\Phi_k(\mathfrak{g}+\langle ik\rangle)-\Phi_j(\mathfrak{g}+\langle ij\rangle)}{\Phi_k(\mathfrak{g}+\langle ik\rangle)\Phi_j(\mathfrak{g}+\langle ij\rangle)}>\frac{d_{ik}(\mathfrak{g})\Phi_k(\mathfrak{g})-d_{ij}(\mathfrak{g})\Phi_j(\mathfrak{g})}{[d_{ik}(\mathfrak{g})\Phi_k(\mathfrak{g})][d_{ij}(\mathfrak{g})\Phi_j(\mathfrak{g})]}$.
\end{proof}
\begin{theorem}
Suppose $i, j, k$ are distinct agents in $\mathfrak{g}$. With respect to the local resource availability,  
\begin{enumerate}
    \item $i$ prefers to add a link with $j$ over $k$ if $j$'s closeness in $\mathfrak{g}+\langle ij \rangle$ is less than that of $k$ in $\mathfrak{g}+\langle ik \rangle$, given that, in $\mathfrak{g}$, both $j$ and $k$ have the same closeness and are at the same distance from $i$. 
    
    \item $i$ prefers to add a link with $j$ over $k$ if, in $\mathfrak{g}$,  $j$'s distance from $i$ is more than that of $k$ from $i$, given that, $j$ and $k$ have the same closeness in $\mathfrak{g}+\langle ij \rangle$ and $\mathfrak{g}+\langle ik \rangle$, respectively, 
as well as in $\mathfrak{g}$.
    
    \item $i$ prefers to add a link with $j$ over $k$ if $j$'s closeness is less than that of $k$ in $\mathfrak{g}$, given that, $j$ and $k$ have the same closeness in $\mathfrak{g}+\langle ij \rangle$ and $\mathfrak{g}+\langle ik \rangle$, respectively, and they are at the same distance from $i$ in $\mathfrak{g}$.

\end{enumerate}
\end{theorem}
\begin{proof}
Follows from Lemma \ref{lemma-choice-max-diff-local}. \\
\end{proof}

\section{Conclusion}
The focus of this study is on endogenous social cloud formation. In this social cloud setting, agents build their resource sharing network to maximize their utility and perform closeness-based conditional resource sharing. The aim of this study is to fill the research gap in the social cloud literature by analyzing the impact of link addition between a pair of agents on their resource availability and that of others. This study provides a theoretical investigation of the impact of link addition on agents' resource availability. It further studies choice modelling, which captures the preferences over agents in link addition. In other words, this study focuses on the agent(s) in the network that is (are) preferred by other agents for link addition, so as to maximize their utility (in terms of their resource availability). 

Our approach of understanding spillover effect (that is, the impact of link addition between a pair of agents on others' resource availability) is different from the existing approaches discussed in the social cloud literature. 
We analyze the relation between closeness and spillover. We show that, for an agent to observe positive spillover effect, it is necessary, but not sufficient, to increase its closeness. Our study does not find a reason for this. This is one of the limitations of the study. Despite this limitation, our study has several implications and applications. 

Our study provides considerable insight into resource availability in endogenous social cloud setting. This analysis will help in defining resource sharing and allocation policy framework in endogenous social cloud systems like BuddyBackup. For example, in these settings, the policy making may include the recommendation of friends (as backup partners) in link addition to avoid negative spillover effect. If an agent wants to select another agent for data backup, then the policy makers can recommend a backup partner who is far from the agent so that negative spillover can be reduced. \\


\section*{\small Declaration of Competing Interest}
The authors declare that they have no known competing financial interests or personal relationships that could have appeared to influence the work reported in this paper.
\section*{\small Conflicts of Interest}
The authors declare that they have no conflict of interest.
%
%


\begin{thebibliography}{19}
\expandafter\ifx\csname natexlab\endcsname\relax\def\natexlab#1{#1}\fi
\providecommand{\url}[1]{\texttt{#1}}
\providecommand{\href}[2]{#2}
\providecommand{\path}[1]{#1}
\providecommand{\DOIprefix}{doi:}
\providecommand{\ArXivprefix}{arXiv:}
\providecommand{\URLprefix}{URL: }
\providecommand{\Pubmedprefix}{pmid:}
\providecommand{\doi}[1]{\href{http://dx.doi.org/#1}{\path{#1}}}
\providecommand{\Pubmed}[1]{\href{pmid:#1}{\path{#1}}}
\providecommand{\bibinfo}[2]{#2}
\ifx\xfnm\relax \def\xfnm[#1]{\unskip,\space#1}\fi
\bibitem[{Boldi and Vigna(2014)}]{Boldi-axioms-centrality}
\bibinfo{author}{Boldi, P.}, \bibinfo{author}{Vigna, S.}, \bibinfo{year}{2014}.
\newblock \bibinfo{title}{Axioms for centrality}.
\newblock \bibinfo{journal}{Internet Math.} \bibinfo{volume}{10},
  \bibinfo{pages}{222--262}.
\bibitem[{Breiger and Pattison(1986)}]{BREIGER1986215}
\bibinfo{author}{Breiger, R.L.}, \bibinfo{author}{Pattison, P.E.},
  \bibinfo{year}{1986}.
\newblock \bibinfo{title}{Cumulated social roles: The duality of persons and
  their algebras}.
\newblock \bibinfo{journal}{Soc. Netw.} \bibinfo{volume}{8},
  \bibinfo{pages}{215 -- 256}.
\bibitem[{Caton et~al.(2012)Caton, Dukat, Grenz, Haas, Pfadenhauer and
  Weinhardt}]{trustfoundations}
\bibinfo{author}{Caton, S.}, \bibinfo{author}{Dukat, C.},
  \bibinfo{author}{Grenz, T.}, \bibinfo{author}{Haas, C.},
  \bibinfo{author}{Pfadenhauer, M.}, \bibinfo{author}{Weinhardt, C.},
  \bibinfo{year}{2012}.
\newblock \bibinfo{title}{Foundations of trust: Contextualising trust in social
  clouds}, in: \bibinfo{booktitle}{Second International Conference on Cloud and
  Green Computing (CGC), 2012}, pp. \bibinfo{pages}{424--429}.
\bibitem[{Caton et~al.(2014)Caton, Haas, Chard, Bubendorfer and
  Rana}]{socialcloud-as-communitycloud}
\bibinfo{author}{Caton, S.}, \bibinfo{author}{Haas, C.},
  \bibinfo{author}{Chard, K.}, \bibinfo{author}{Bubendorfer, K.},
  \bibinfo{author}{Rana, O.}, \bibinfo{year}{2014}.
\newblock \bibinfo{title}{A social compute cloud: Allocating and sharing
  infrastructure resources via social networks}.
\newblock \bibinfo{journal}{IEEE Trans. Serv. Comput.}
  \bibinfo{volume}{7}, \bibinfo{pages}{359--372}.
\bibitem[{Chard et~al.(2012)Chard, Bubendorfer, Caton and Rana}]{transaction}
\bibinfo{author}{Chard, K.}, \bibinfo{author}{Bubendorfer, K.},
  \bibinfo{author}{Caton, S.}, \bibinfo{author}{Rana, O.},
  \bibinfo{year}{2012}.
\newblock \bibinfo{title}{Social cloud computing: A vision for socially
  motivated resource sharing}.
\newblock \bibinfo{journal}{IEEE Trans. Serv. Comput.}
  \bibinfo{volume}{5}, \bibinfo{pages}{551--563}.
\bibitem[{Chard et~al.(2010)Chard, Caton, Rana and Bubendorfer}]{conference}
\bibinfo{author}{Chard, K.}, \bibinfo{author}{Caton, S.},
  \bibinfo{author}{Rana, O.}, \bibinfo{author}{Bubendorfer, K.},
  \bibinfo{year}{2010}.
\newblock \bibinfo{title}{Social cloud: Cloud computing in social networks},
  in: \bibinfo{booktitle}{3rd IEEE International Conference on Cloud Computing
  (CLOUD)}, pp. \bibinfo{pages}{99--106}.
\bibitem[{Chard et~al.(2015)Chard, Caton, Rana and
  Bubendorfer}]{Chard-Clouds-Retrospective-2015}
\bibinfo{author}{Chard, K.}, \bibinfo{author}{Caton, S.},
  \bibinfo{author}{Rana, O.}, \bibinfo{author}{Bubendorfer, K.},
  \bibinfo{year}{2015}.
\newblock \bibinfo{title}{Social clouds: A retrospective}.
\newblock \bibinfo{journal}{IEEE Cloud Comput.} \bibinfo{volume}{2},
  \bibinfo{pages}{30--40}.
\bibitem[{Haas et~al.(2013)Haas, Caton, Chard and
  Weinhardt}]{cooperativesocialcloud}
\bibinfo{author}{Haas, C.}, \bibinfo{author}{Caton, S.},
  \bibinfo{author}{Chard, K.}, \bibinfo{author}{Weinhardt, C.},
  \bibinfo{year}{2013}.
\newblock \bibinfo{title}{Co-operative infrastructures: An economic model for
  providing infrastructures for social cloud computing}, in:
  \bibinfo{booktitle}{46th Hawaii International Conference on System Sciences
  (HICSS), 2013}, pp. \bibinfo{pages}{729--738}.
\bibitem[{Jackson(2008)}]{jackson-book}
\bibinfo{author}{Jackson, M.O.}, \bibinfo{year}{2008}.
\newblock \bibinfo{title}{Social and Economic Networks}.
\newblock \bibinfo{publisher}{Princeton University Press}.
\bibitem[{Mane et~al.(2014)Mane, Krishnamurthy and
  Ahuja}]{social-cloud-gamenets}
\bibinfo{author}{Mane, P.}, \bibinfo{author}{Krishnamurthy, N.},
  \bibinfo{author}{Ahuja, K.}, \bibinfo{year}{2014}.
\newblock \bibinfo{title}{Externalities and stability in social cloud}, in:
  \bibinfo{booktitle}{5th International Conference on Game Theory for Networks (GAMENETS), 2014}, pp. \bibinfo{pages}{1--6}.
\bibitem[{Mane et~al.(2019)Mane, Ahuja and Krishnamurthy}]{Pramod-Games}
\bibinfo{author}{Mane, P.~C.}, \bibinfo{author}{Krishnamurthy, N.},
  \bibinfo{author}{Ahuja, K.}, \bibinfo{year}{2019}.
\newblock \bibinfo{title}{Formation of stable and efficient social storage cloud}.
\newblock \bibinfo{journal}{Games} \bibinfo{volume}{10}, \bibinfo{article no.}{44}.
\bibitem[{Mane et~al.(2020a)Mane, Ahuja and Krishnamurthy}]{Pramod-AEL}
\bibinfo{author}{Mane, P.C.}, \bibinfo{author}{Ahuja, K.},
  \bibinfo{author}{Krishnamurthy, N.}, \bibinfo{year}{2020}.
\newblock \bibinfo{title}{Externalities in socially-based resource sharing network}.
\newblock \bibinfo{journal}{Appl. Econ. Lett.} \bibinfo{volume}{27}, \bibinfo{pages}{1404--1408}.
\bibitem[{Mane et~al.(2020b)Mane, Ahuja and Krishnamurthy}]{Pramod-ANOR}
\bibinfo{author}{Mane, P.C.}, \bibinfo{author}{Ahuja, K.},
  \bibinfo{author}{Krishnamurthy, N.}, \bibinfo{year}{2020}.
\newblock \bibinfo{title}{Stability, efficiency, and contentedness of social storage networks}.
\newblock \bibinfo{journal}{Ann. Oper. Res.} \bibinfo{volume}{287}, \bibinfo{pages}{811--842}.
\bibitem[{Marchiori and Latora(2000)}]{marchiori2000harmony}
\bibinfo{author}{Marchiori, M.}, \bibinfo{author}{Latora, V.},
  \bibinfo{year}{2000}.
\newblock \bibinfo{title}{Harmony in the small-world}.
\newblock \bibinfo{journal}{Physica A} \bibinfo{volume}{285}, \bibinfo{pages}{539--546}.
\bibitem[{Mohaisen et~al.(2014)Mohaisen, Tran, Chandra and Kim}]{distributed}
\bibinfo{author}{Mohaisen, A.}, \bibinfo{author}{Tran, H.},
  \bibinfo{author}{Chandra, A.}, \bibinfo{author}{Kim, Y.},
  \bibinfo{year}{2014}.
\newblock \bibinfo{title}{Trustworthy distributed computing on social networks}.
\newblock \bibinfo{journal}{IEEE Trans. Serv. Comput.}
  \bibinfo{volume}{7}, \bibinfo{pages}{333 -- 345}.
\bibitem[{Moscibroda et~al.(2011)Moscibroda, Schmid and
  Wattenhofer}]{Stefan2011}
\bibinfo{author}{Moscibroda, T.}, \bibinfo{author}{Schmid, S.},
  \bibinfo{author}{Wattenhofer, R.}, \bibinfo{year}{2011}.
\newblock \bibinfo{title}{Topological implications of selfish neighbor
  selection in unstructured peer-to-peer networks}.
\newblock \bibinfo{journal}{Algorithmica} \bibinfo{volume}{61},
  \bibinfo{pages}{419--446}.
\bibitem[{Opsahl et~al.(2010)Opsahl, Agneessens and Skvoretz}]{opsahl}
\bibinfo{author}{Opsahl, T.}, \bibinfo{author}{Agneessens, F.},
  \bibinfo{author}{Skvoretz, J.}, \bibinfo{year}{2010}.
\newblock \bibinfo{title}{Node centrality in weighted networks: Generalizing
  degree and shortest paths}.
\newblock \bibinfo{journal}{Soc. Netw.} \bibinfo{volume}{32},
  \bibinfo{pages}{245--251}.
\bibitem[{Punceva et~al.(2015)Punceva, Rodero, Parashar, Rana and
  Petri}]{incentivesocial}
\bibinfo{author}{Punceva, M.}, \bibinfo{author}{Rodero, I.},
  \bibinfo{author}{Parashar, M.}, \bibinfo{author}{Rana, O.F.},
  \bibinfo{author}{Petri, I.}, \bibinfo{year}{2015}.
\newblock \bibinfo{title}{Incentivising resource sharing in social clouds}.
\newblock \bibinfo{journal}{Concurr. Comput.} \bibinfo{volume}{27}, \bibinfo{pages}{1483--1497}.
\bibitem[{Zhang and van~der Schaar(2013)}]{joballocationscoialcloud}
\bibinfo{author}{Zhang, Y.}, \bibinfo{author}{van~der Schaar, M.},
  \bibinfo{year}{2013}.
\newblock \bibinfo{title}{Incentive provision and job allocation in social
  cloud systems}.
\newblock \bibinfo{journal}{IEEE J. Sel. Area Comm.}
  \bibinfo{volume}{31}, \bibinfo{pages}{607--617}.
\bibitem[{Zuo and Iamnitchi(2016)}]{Zuo-2016-P2P}
\bibinfo{author}{Zuo, X.}, \bibinfo{author}{Iamnitchi, A.},
  \bibinfo{year}{2016}.
\newblock \bibinfo{title}{A survey of socially aware peer-to-peer systems}.
\newblock \bibinfo{journal}{ACM Comput. Surv.} \bibinfo{volume}{49},
  \bibinfo{pages}{9:1--9:28}.

\end{thebibliography}

\end{document}